\documentclass[11pt,letterpaper]{article}
\usepackage{graphicx} 
\usepackage[table]{xcolor}
\usepackage{amsmath}
\usepackage{amssymb}
\usepackage{mathrsfs}
\usepackage{amsthm}
\usepackage{thm-restate}
\usepackage{mathtools}
\usepackage[margin=1in]{geometry}
\usepackage{hyperref}
\usepackage{cleveref}
\usepackage{algorithm}
\usepackage[noend]{algpseudocode}
\usepackage[style=numeric-comp,maxnames=99,minnames=99,sortcites=true]{biblatex}
\usepackage{tikz}
\usetikzlibrary{decorations.pathreplacing}

\usepackage{authblk}
\usepackage{orcidlink}
\newcommand{\myorcid}[1]{{\small \orcidlink{#1}}}

\theoremstyle{plain}
\newtheorem{theorem}{Theorem}[section]
\newtheorem{lemma}[theorem]{Lemma}

\newtheorem{corollary}[theorem]{Corollary}

\newtheorem{observation}[theorem]{Observation}
\newtheorem{definition}[theorem]{Definition}
\newtheorem{claim*}{Claim}

\newtheorem{question}{Question}

\crefname{instance}{Instance}{Instances}
\Crefname{instance}{Instance}{Instances}

\newcommand{\poly}{\ensuremath{\operatorname{poly}}}
\newcommand{\eps}{\ensuremath{\varepsilon}}

\newif\ifdraftmode
\draftmodetrue

\Crefname{property}{property}{properties}
\Crefname{property}{Property}{Properties}

\newcommand{\degree}[2]{\operatorname{deg_{#1}}(#2)}
\newcommand{\goodnodes}[1]{#1_{\operatorname{good}}}
\newcommand{\badnodes}[1]{#1_{\operatorname{bad}}}
\newcommand{\neutralnodes}[1]{#1_{\operatorname{neutral}}}
\newcommand{\net}[3]{$(#1,#2,#3)$-net}
\newcommand{\weaknet}[3]{weak $(#1,#2,#3)$-net}
\newcommand{\ecc}[2]{$\operatorname{ecc}_{#1}(#2)$}
\newcommand{\rad}[1]{$\operatorname{rad}(#1)$}
\newcommand{\Deltaball}[0]{\Delta}

\newcommand{\algoname}[1]{\textnormal{{\textsc{#1}}}}

\newcommand{\EdgeColoring}{\algoname{EdgeColoring}}

\newcommand{\NodePicking}{\algoname{NodeSelection}}
\newcommand{\SparseNodePicking}{\algoname{SparseNodeSelection}}
\newcommand{\BallCreation}{\algoname{BallCreation}}

\title{Fast Distributed $(\Delta + 1)$-Edge Coloring in Bipartite Graphs}
\author{Yann Bourreau \myorcid{0009-0001-1819-8348}}
\author{Sebastian Brandt \myorcid{0000-0001-5393-6636}}
\author{Ananth Narayanan \myorcid{0009-0002-6137-4025}}
\affil{CISPA Helmholtz Center for Information Security, Saarbrücken, Germany}

\begin{document}

\maketitle

\begin{abstract}
    We show that, on bipartite $n$-vertex graphs in the LOCAL model of distributed computation, $(\Delta + 1)$-edge coloring can be solved deterministically in $O(f(\Delta) \log n)$ rounds for some function $f$.
    In particular this implies a deterministic $O(\log n)$-round upper bound for $(\Delta + 1)$-edge coloring on bipartite constant-degree graphs.
    Our results improve on the previous state-of-the-art bounds (which hold already on general graphs) whose $n$-dependencies are $\widetilde{O}(\log^5 n)$ rounds deterministically and $\widetilde{O}(\log^2 n)$ rounds randomized [Bernshteyn and Dhawan, J.\ Comb.\ Theory 2025].

    In fact, we obtain the aforementioned results already for a more restrictive version of $(\Delta + 1)$-edge coloring that corresponds to a \emph{local-palette} strengthening of Vizing's theorem shown recently to hold [Christiansen, STOC 2023].
    In this version, every edge $\{u, v\}$ has to be colored with a color from $\{1, \dots, \max(\deg(u), \deg(v)) + 1\}$.
    For all of our results, the vertices do not need to know on which side of the bipartition they lie---the structural promise of bipartiteness suffices. 
\end{abstract}

\section{Introduction}\label{sec:introduction}
Edge coloring is one of the classical problems in graph theory, studied for many decades.
One of the cornerstone results on edge coloring, Vizing's theorem \cite{vizing1964}, states that every simple graph of maximum degree $\Delta$ admits a proper edge coloring with $\Delta + 1$ colors, and this is best possible in general: odd cycles, for example, require three colors although $\Delta = 2$.
Since Vizing's discovery, $(\Delta + 1)$-edge coloring has been examined across many models of computation, with substantial progress especially in recent years (see, e.g., \cite{bernshteyn2022fast,christiansen2023power,bhattacharya2024faster,bhattacharya2024nibbling,bernshteyn2025fast,bhattacharya2025even,assadi2025faster,assadi2025vizing,christiansen2026deterministic,assadi2026vizing} for the most recent progress in sequential, dynamic, and distributed algorithms).
In this work, we study the deterministic complexity of $(\Delta + 1)$-edge coloring in the standard LOCAL model of distributed computation.

\paragraph{The LOCAL model.}
In the LOCAL model \cite{linial1992locality}, the input graph is considered as a communication network with the vertices (or nodes) of the graph representing computational entities: computation proceeds in synchronous rounds, and in each round every vertex may send an arbitrarily large message to each neighbor, receive the messages sent to it, and perform arbitrary local computation.
Each vertex has to terminate at some point and then output its local part of the global solution to the considered problem; for $(\Delta + 1)$-edge coloring, each vertex has to output a color from $\{ 1, \dots, \Delta + 1 \}$ for each of its incident edges such that endpoints of the same edge agree on the color of the edge.
The runtime of a distributed algorithm governing the behavior of the vertices is simply the number of rounds until the last vertex has terminated.
For a full definition of the LOCAL model, we refer the reader to \Cref{sec:prelims}.

When considering bipartite graphs in the LOCAL model, depending on the definition of the model in the bipartite setting, vertices may be aware on which side of the bipartition they lie or may not have this information.
To distinguish between these two cases, we say that the vertices of a bipartite graph \emph{know (resp.\ do not know) the bipartition} if each vertex is initially given one bit telling it on which side of a fixed bipartition it lies (resp.\ does not receive such information).
Throughout the paper, $n$ denotes the number of vertices of the input graph and $\Delta$ denotes its maximum degree.

\paragraph{Complexity bounds for $(\Delta + 1)$-edge coloring.}
As recently as 2020, it was open whether Vizing's theorem admits an implementation in the LOCAL model with runtime polylogarithmic in the global parameter $n$, even if an additional multiplicative dependence on the local parameter $\Delta$ is allowed, i.e., whether $(\Delta + 1)$-edge coloring can be solved in $O(f(\Delta) \log^{O(1)} n)$ rounds for some function $f$.
This question was answered in a breakthrough result by Bernshteyn who showed that $(\Delta + 1)$-edge coloring admits a deterministic algorithm with a runtime of $\widetilde{O}(f(\Delta) \log^{11} n)$ rounds\footnote{We use $\widetilde{O}(\cdot)$ to hide a $\poly(\log \log n)$ factor. We also note that the exponent $11$ was not given explicitly in \cite{bernshteyn2022fast} but observed later in \cite{bernshteyn2025fast}.} \cite{bernshteyn2022fast}.
Christiansen subsequently improved the bound to $\widetilde{O}(f(\Delta) \log^6 n)$ rounds via multi-step Vizing chains \cite{christiansen2026deterministic}, and Bernshteyn and Dhawan improved it further to the current deterministic state-of-the-art bound of $\widetilde{O}(f(\Delta) \log^5 n)$ rounds \cite{bernshteyn2025fast}.
In the same paper, Bernshteyn and Dhawan also obtained the randomized state of the art by providing a randomized $O(f(\Delta) \log^2 n)$-round algorithm.

On the lower bound side, it is known from \cite{chang2019distributed} that, already on constant-degree graphs, $\Omega(\log n)$ rounds are required for deterministic algorithms and $\Omega(\log \log n)$ rounds for randomized algorithms, leaving the question about the true complexity of $(\Delta + 1)$-edge coloring wide open.
In particular, the current state of deterministic upper and lower bounds raises the following fundamental question.

\begin{question}\label{que:logarithmic}
Can $(\Delta+1)$-edge coloring be solved deterministically in $O(f(\Delta)\log n)$ rounds, or is a superlogarithmic dependency on $n$ necessary?
\end{question}
An affirmative answer to \Cref{que:logarithmic} would also resolve
an open question by Bernshteyn and Dhawan when restricted to bipartite graphs, asking whether there exists a deterministic or randomized $O(\log n)$-round algorithm for $(\Delta + 1)$-edge coloring on constant-degree graphs \cite[Question 1.9]{bernshteyn2025fast}.

\paragraph{Superlogarithmic lower bounds.}
\Cref{que:logarithmic} is also significant from the broader perspective of proving LOCAL lower bounds for locally checkable problems as we discuss next.
Locally checkable labeling problems are problems that can be defined via local constraints---if around each vertex or edge the local constraints are satisfied by the output, the global solution is considered to be correct; otherwise it is considered to be incorrect.
A large variety of problems that have been studied in the LOCAL model are locally checkable problems, including the central symmetry-breaking problems such as standard coloring problems, maximal independent set, or maximal matching.
However, despite the ubiquity of locally checkable problems, the main general-purpose lower bound techniques for such problems, including round elimination \cite{brandt2016lower,brandt2019automatic} and self-reduction \cite{khoury2025round,balliu2026new}, the KMW lower bound framework \cite{kuhn2016local,balliu2023node}, and Marks' technique \cite{marks2016determinacy,brandt2022local}, are, in their present forms, inherently unable to prove deterministic lower bounds strictly above $\Omega(\log n)$ rounds, except when the considered problem is already global (i.e., has an $\Omega(n)$-round lower bound) on paths or cycles.

We note that locally checkable problems exhibiting superlogarithmic lower bounds are known but, to the best of our knowledge, these are either highly artificial problems engineered specifically for the purpose of exhibiting a given complexity \cite{chang2019time,balliu2018new,balliu2021almost} or problems that require a high polynomial number of rounds to be solved (such as global problems like $2$-coloring a path, which require $\Omega(n)$ rounds, or locally optimal cut, which requires $\Omega(\sqrt{n})$ rounds \cite{balliu2026distributed}).
In particular, there is no natural locally checkable problem that is known to be solvable in subpolynomial time but not in $O(\log n)$ rounds.
An even stronger dichotomy appears when restricting attention to the class of problems on constant-degree graphs, which has been in the center of attention for many years (see, e.g., \cite{balliu2018new,balliu2021almost,chang2019exponential,chang2019time,grunau2022landscape} for research on the distributed complexity landscape of so-called LCL problems, which are essentially locally checkable problems on constant-degree graphs): we are not aware of any natural locally checkable problem with a complexity known to liein both $\omega(\log n)$ and $o(n)$.

Identifying a problem as described above is a major open question (which likely would require coming up with a broadly applicable new lower bound technique); even proving, say, an $\Omega(\log n \cdot \log^* n)$-round lower bound for a natural non-global locally checkable problem on constant-degree graphs would be a major breakthrough.
Due to its restrictive nature, $(\Delta + 1)$-edge coloring is arguably one of the most promising candidates for such a superlogarithmic lower bound: after all, the task in the $(\Delta + 1)$-edge coloring problem is to color with the existentially optimal number of colors, which is far below both the greedy threshold of $2\Delta - 1$ colors and the current state-of-the-art number of colors known to admit a deterministic $O(\log n)$-round edge coloring algorithm on constant-degree graphs, which is $3\Delta/2$ \cite{brandt2025locality}.
In fact, even the arguably much less restrictive $\Delta$-vertex coloring problem, which allows only one color less than the greedily solvable $(\Delta + 1)$-vertex coloring problem, was considered to be a candidate for exhibiting a superlogarithmic lower bound on constant-degree graphs (see \cite[Section 1.1]{bourreau2025faster}) until a deterministic $O(\log n)$-round algorithm was found recently \cite{bourreau2026faster}.

\subsection{Our Contributions}\label{sec:contributions}
Our main contribution is the following theorem.
\begin{restatable}{theorem}{MainEdgeColoringTheorem}\label{thm:edge-col-main}
    There exists a function $f$ such that, on bipartite graphs, a $(\Delta + 1)$-edge coloring can be computed deterministically in $O(f(\Delta) \log n)$ rounds. This holds even when the bipartition is not known to the vertices.
\end{restatable}

\Cref{thm:edge-col-main} answers \Cref{que:logarithmic} affirmatively for bipartite graphs.
It improves the previous deterministic state of the art for $(\Delta+1)$-edge coloring on this graph class from $\widetilde O(f(\Delta)\log^5 n)$ rounds \cite{bernshteyn2025fast} to $O(f(\Delta)\log n)$ rounds, and it even improves over the previous randomized $O(f(\Delta)\log^2 n)$-round bound \cite{bernshteyn2025fast} while remaining deterministic.
In terms of the number of colors for which bipartite edge-coloring is known to be solvable in $O(f(\Delta) \log n)$ rounds, \Cref{thm:edge-col-main} improves over the previous state of the art of $3\Delta/2$ colors \cite{brandt2025locality}.

We note that in all previous $(\Delta + 1)$-edge coloring algorithms whose runtimes have a polylogarithmic dependency on $n$ \cite{bernshteyn2022fast,christiansen2023power,bernshteyn2025fast} the dependency on $\Delta$ is polynomial, while in our algorithm the runtime is considerably superpolynomial.
However, we also did not attempt to minimize the dependency on $\Delta$, which is in line with the common\footnote{For instance, the previous state of the art by Bernshteyn and Dhawan \cite{bernshteyn2025fast} has a dependency of $\Delta^{84}$ and $\Delta^{36}$ on $\Delta$ in the deterministic and randomized algorithms, respectively. In fact, the title of the paper ``Fast algorithms for Vizing's theorem on bounded degree graphs'' already underlines the emphasis of current research on the $n$-dependency.} focus on the dependency on $n$, leaving improving the exact dependency on the local parameter $\Delta$ to future work.
In particular, in light of the deterministic $\Omega(\log n)$-round lower bound proved in \cite{chang2019distributed} we obtain the following important corollary of \Cref{thm:edge-col-main}, which also answers an open question by Bernshteyn and Dhawan \cite[Question 1.9]{bernshteyn2025fast} when restricted to bipartite graphs.

\begin{corollary}\label{cor:constant-degree}
    On constant-degree bipartite graphs, with or without the bipartition being known to the vertices, the deterministic complexity of $(\Delta + 1)$-edge coloring is $\Theta(\log n)$ rounds.
\end{corollary}
The tightness of our $O(\log n)$-round upper bound on constant-degree bipartite graphs (which does not require that the vertices know the bipartition) follows from the fact that the aforementioned $\Omega(\log n)$-round lower bound holds already on constant-degree bipartite graphs where the bipartition is known to the vertices.

From the general perspective of proving lower bounds, \Cref{thm:edge-col-main} and \Cref{cor:constant-degree} show that $(\Delta + 1)$-edge coloring, perhaps surprisingly, cannot be used for proving superlogarithmic lower bounds on constant-degree graphs, leaving open the question of finding another promising (non-artificial, non-global) candidate problem for exhibiting such a lower bound.
Our results may also be surprising in the sense that bipartiteness is required only as a structural promise on the input instance: they show that, for $(\Delta + 1)$-edge coloring, the existence of a bipartition can be algorithmically exploited without exposing that bipartition to the vertices.

We see our results as a major step towards a full resolution of \Cref{que:logarithmic}: solving the bipartite case is often a first step toward the general case (see, e.g., \cite{hanckowiak1998distributed,lotker2015improved,ahmadi2018distributed,fischer2020improved,khoury2025randomized}), in particular in the context of matching problems, which play an essential role in our overall approach.

\paragraph{Stronger local-palette statements.}
While \Cref{thm:edge-col-main} and \Cref{cor:constant-degree} can be considered as our main results since they are stated for standard $(\Delta + 1)$-edge coloring, we in fact prove stronger statements using a \emph{local} color palette where the set of allowed colors for an edge is determined by the degrees of its endpoints (instead of by the global maximum degree).
The relevant existential statement, shown in \cite{christiansen2023power}, is a local strengthening of Vizing's theorem: every graph admits an edge coloring in which each edge $\{u,v\}$ receives a color from the palette
\[
    \{1,\ldots,\max(\deg(u),\deg(v))+1\}.
\]

We show that our algorithm can also be used to compute these more restrictive edge colorings, thereby obtaining the following two statements.

\begin{restatable}{theorem}{LocalPaletteTheorem}\label{thm:stronger-version}
    On bipartite graphs (even without the bipartition being known to the vertices), a $(\Delta + 1)$-edge coloring where each edge $\{ u, v\}$ receives a color from $\{ 1, \dots, \max(\deg(u), \deg(v)) + 1 \}$ can be computed deterministically in $O(f(\Delta) \log n)$ rounds, for some function $f$.
\end{restatable}

\begin{corollary}\label{cor:stronger-constant-degree}
    On constant-degree bipartite graphs (with or without the bipartition being known to the vertices), the deterministic complexity of $(\Delta + 1)$-edge coloring where each edge $\{ u, v\}$ receives a color from $\{ 1, \dots, \max(\deg(u), \deg(v)) + 1 \}$ is $\Theta(\log n)$ rounds.
\end{corollary}

We note that all polylogarithmic-in-$n$ algorithms for $(\Delta + 1)$-edge coloring \cite{bernshteyn2022fast,christiansen2023power,bernshteyn2025fast} do not seem to extend to the local-palette version of Vizing's theorem (in any known way): as such, \Cref{thm:stronger-version} and \Cref{cor:stronger-constant-degree} provide the first results for this problem with a polylogarithmic dependency on $n$.

\subsection{Our Approach}\label{sec:approach}
In this section, we provide an informal high-level overview of our approach.
For simplicity, we describe the algorithm for standard $(\Delta + 1)$-edge coloring; the local-palette version can be obtained in a similar manner.

\paragraph{Comparison with prior work.}
Our algorithm follows a very different route from the Vizing-chain-based approaches used in the previous algorithms for $(\Delta+1)$-edge coloring achieving a polylogarithmic dependency on $n$ \cite{bernshteyn2022fast,christiansen2023power,bernshteyn2025fast}.
In those works, the edge coloring is constructed in iterations roughly as follows: In each iteration a subset $S$ of uncolored edges is colored by finding, for each such edge $e$, a ``chain'' of edges, starting at $e$, that can be recolored in a way that guarantees that there exists a color $c$ that no edge neighboring $e$ is colored with.
After recoloring the chain, $e$ is then colored with color $c$.
To avoid issues coming from interference between chains for different edges in $S$, the set $S$ is selected so that these chains do not intersect, which requires a lot of care.

In contrast, in our approach, we build the coloring one color class at a time: first we select a subset of edges that receive color $\Delta + 1$, then a subset of edges that receive color $\Delta$, and so on. We guarantee that we color ``sufficiently many'' edges with each color so that after iterating through all colors no edge remains uncolored.
In particular, once an edge receives a color, this edge will never be recolored later in the algorithm.

\paragraph{Problematic vertices.}
If, in each iteration, we could find a perfect matching on the set of uncolored edges (and then assign the edges in this matching the currently considered color), this general approach would even yield a $\Delta$-edge coloring.
However, even with $\Delta + 1$ colors available, for high-degree vertices there is not much ``slack'' that we can exploit: for instance, if, for a vertex $v$ of degree $\Delta$, no incident edge is colored in some given iteration, then in all other iterations (exactly) one edge incident to $v$ must be colored (as all $\Delta$ incident edges must be colored at the end of the $\Delta + 1$ iterations and no two neighboring edges can be colored in the same iteration).
To keep track of the vertices that already exhausted their slack throughout our algorithm, we essentially define a set of problematic vertices that grows across the $\Delta + 1$ iterations: as soon as, for a vertex $v$, the number of incident uncolored edges matches the number of remaining iterations for the first time, $v$ becomes ``problematic'', indicating that in every future iteration an incident edge has to be colored.
Initially, no vertex is problematic as the number of remaining iterations is $\Delta + 1$.

As, for an overall runtime of $O(f(\Delta)\log n)$ rounds, we may spend $O(f(\Delta)\log n)$ rounds in each of the $\Delta + 1$ iterations, we are left with the following challenge: in each iteration, compute a large matching in $O(f(\Delta)\log n)$ rounds that at the very least matches all problematic vertices (which defines a set of pairwise non-adjacent edges that then can be colored with a common color and guarantees that for each problematic vertex an incident edge is colored).
Note that the larger the set of problematic vertices becomes, the closer to a perfect matching the matchings we compute in subsequent iterations have to become, which considerably complicates our task of finding such matchings fast.

\paragraph{Keeping problematic vertices isolated.} A key ingredient of our approach is to preserve the property that any two problematic vertices are far from each other.
Perhaps surprisingly, we show that this can be guaranteed across essentially all iterations while paying only a logarithmic dependency on $n$ in the runtime.
More specifically, we maintain a lower bound on the minimum pairwise distance between problematic vertices that gracefully degrades across iterations; after the first iteration, problematic vertices are guaranteed to have a very large minimum pairwise distance and this distance guarantee decreases in each iteration.
However, in each iteration, this lower bound on the pairwise distance of problematic vertices can be expressed purely as a function of $\Delta$.

In the following we explain how we compute a matching in a given iteration without decreasing the aforementioned guarantee on the minimum pairwise distance between problematic vertices too much.
Consider some iteration $i$ and assume that at the beginning of the iteration the problematic vertices are guaranteed to have a pairwise distance of at least $d$.
Our task is to find a matching in the graph induced by the uncolored edges that matches at least all problematic vertices but ideally many more (so that the number of problematic vertices increases as little as possible).

\paragraph{Computing saturating matchings.} To this end, one recent powerful technique that we would like to make use of is the deterministic ``hypergraph sinkless orientation'' algorithm from \cite{brandt2025locality}, which has already found a number of applications for different problems \cite{brandt2025locality,jakob2025towards,bourreau2026faster,maus2026fast,jakob2026towards}.
This algorithm is developed for bipartite graphs where the minimum degree $a$ on one side of the bipartition (say, $X$) is strictly larger than the maximum degree $b$ on the other side; it computes a matching that matches every vertex on side $X$ (which we call an \emph{$X$-saturating matching}).
 The runtime of the algorithm is $O(\log_{a/b} n)$ rounds.

\paragraph{Obstacles.} Unfortunately, the graphs at the beginning of our iterations do not have the aforementioned properties required for the algorithm to work: while our graphs are bipartite, the vertices do not know on which side of the bipartition they lie, and the graphs do not satisfy the aforementioned max-min-degree relation in general.
In fact, if our input graph is a $\Delta$-regular (bipartite) graph (which, intuitively, is the worst possible instance as it provides the least slack in the aforementioned sense), the average degree on both sides of the bipartition is initially the same and will always remain the same as each edge that is ``removed'' due to becoming colored reduces the sum of degrees on both sides by $1$.
In that case, the graphs at the beginning of the iterations will be far away from anything satisfying the aforementioned max-min-degree relation.

However, despite these issues that seem to rule out approaches using the algorithm from \cite{brandt2025locality}, we develop an approach that is at least partially inspired by that algorithm.
Unsurprisingly, our approach requires a number of novel ingredients that need to be carefully combined.

\paragraph{An extended algorithm for computing saturating matchings.}
One ingredient at the heart of our approach is an extension of the algorithm from \cite{brandt2025locality} that works on all graphs satisfying a property that is much weaker than the aforementioned max-min-degree condition.
Very informally, we show that even on bipartite graphs $(X \cup Y, E)$ where most vertices have the same degree (say, $k$), we can find an $X$-saturating matching fast if a certain set of vertices called \emph{good vertices} lies sufficiently ``denser'' in the graph than a certain set of vertices called \emph{bad vertices} (see \Cref{thm:Density-HSO}).
Here, the good vertices are simply the vertices in $X$ with degree $> k$ and the vertices in $Y$ with degree $< k$ whereas the bad vertices are the vertices in $Y$ with degree $> k$.\footnote{There is some asymmetry in the definition as for the graphs we will consider in this context, all vertices in $X$ have degree $\geq k$.}
In a sense, the definitions of good and bad vertices resemble the discussed max-min-degree condition, though in a much weaker form.
Given the impact that \cite{brandt2025locality} had since its publication in 2025, we believe that our extension will be highly useful for developing even faster algorithms for various problems.
The proof of the correctness and runtime of our extended algorithm is inspired by the proof for the algorithm in \cite{brandt2025locality} but becomes substantially more involved due to the much weaker condition.

\paragraph{Applying our algorithm on a virtual graph.}
While our new algorithm is applicable on a considerably larger class of graphs, the graphs at the beginning of the iterations still do not fall into this class in general, and we still have the issue that vertices are not aware of the bipartition underlying the input graphs.
To solve the latter issue, we consider the bipartite double cover of the input graph, which we denote by $B = (X \cup Y, E)$ in the following.
To solve the former issue, we ``duplicate'' a carefully chosen set of vertices of $B$ and remove vertices in $X$ with small degree, thereby changing the degrees of vertices of $B$ in a way that guarantees that our new algorithm is applicable on the obtained graph.
The application of the algorithm then yields, in $O(f(\Delta) \log n)$ rounds, a matching that matches all (remaining) vertices in $X$.
Unfortunately, the virtual graph on which we compute this matching is too different from the input graph to transform the computed matching directly into a matching on the input graph that satisfies the severe restrictions regarding size and problematic vertices required to sustain our approach over many iterations.
However, we show that the computed matching is highly useful in a different manner: it gives rise to a well-structured subgraph of the input graph which, in particular, satisfies strong degree guarantees, as detailed next.

\paragraph{Extracting a well-structured subgraph of the input graph.}
Consider the set $E'$ of edges $e'$ of our original input graph $G$ for which at least one ``copy'' of $e'$ in $B$ has been selected into the matching.
Here, an edge $e = \{ u, v \}$ of $B$ is considered as a copy of an edge $e' = \{ u', v' \}$ of $G$ if each of $u$ and $v$ is a copy of or duplicate of a copy of $u$ or $v$ in the bipartite double cover. 
We show that the edges in $E'$ induce a subgraph $G'$ of $G$ that has minimum degree $1$ and maximum degree $3$; moreover, if a vertex $v$ has degree $3$ in $G'$, $v$ is a vertex whose copy in $B$ has been duplicated.
We show further that the vertices to be duplicated can be selected (without compromising the applicability of our algorithm) in a way that guarantees that they are far away from each other and from the problematic vertices.

\paragraph{Computing the desired matching.}
Finally, we find a matching $M'$ on $G'$ (and therefore also on $G$) such that the vertices that remain unmatched are far away from each other and from problematic vertices (in $G$); we achieve this in $O(f(\Delta) \log n)$ rounds by exploiting the specific structure of $G'$ and combining it with a novel application of the sinkless orientation algorithm from \cite{GhaffariS17}.
The matching $M'$ computed in this way now constitutes the set of edges that we color at the end of the current iteration (all with the same color reserved for this iteration).
Since our construction of $M'$ guarantees that the unmatched vertices are sufficiently far from each other and from problematic vertices (though the lower bound on the pairwise distance between these vertices is lower than the lower bound $d$ at the start of the iteration), we obtain that the new set of problematic vertices at the start of the next iteration (which is a subset of the problematic vertices of the current iteration and the unmatched vertices) is again suitably lower bounded by some expression that is a function of $\Delta$.

\paragraph{Success across all iterations.}
Hence, by choosing our parameters carefully, we can guarantee that at the start of every iteration the lower bound guarantee on the minimum pairwise distance between problematic vertices suffices to perform the respective iteration in $O(f(\Delta) \log n)$ rounds, until the graph induced by the uncolored edges has maximum degree $3$ and satisfies that vertices of degree $3$ are still sufficiently far from each other.
As, at this point, still $3$ colors remain, it is easy to edge-color the remaining graph suitably by choosing different colors for any three edges incident to a common vertex of degree $3$ and then fixing the coloring on the remaining paths and cycles in $O(\log^* n)$ rounds.
This completes the high-level description of our algorithm.

\paragraph{Omitted details.}
We note that there are many intricate details of the algorithm and analysis that we omitted for simplicity.
In particular, whenever at the start of an iteration the graph induced by the uncolored edges contains connected components of relatively small diameter, we solve these components simply by brute force, e.g., by collecting the entire component in one vertex, computing an edge coloring with the remaining colors on the component and distributing the solution to all vertices in the component.
The existence of a proper edge coloring of the component with the remaining colors follows from the fact that the number of remaining colors is at least as large as the maximum degree in the component and the well-known result that any bipartite graph is $\Delta$-edge-colorable \cite{Konig1916}.

\subsection{Further Related Work}
\paragraph{Deterministic state-of-the-art bounds depending on the number of colors.}
The greedy threshold for edge coloring is $2\Delta-1$: there is a sequential greedy algorithm with this number of colors that simply iterates through the edges in an arbitrary order and colors each edge with a color that none of its neighboring edges has been assigned.
Due to its greedy nature, $(2\Delta - 1)$-edge coloring is often considered to be a considerably easier problem than $(\Delta + 1)$-edge coloring; this view is supported by the current state-of-the-art runtimes for $(2\Delta - 1)$-edge coloring, which are considerably below what is known (or even possible) for $(\Delta + 1)$-edge coloring.
In particular, $(2\Delta - 1)$-edge coloring can be solved in $O(\log^* n + \log^{12} \Delta)$ \cite{balliu2022distributed}, $O(\log n \log^2 \Delta)$ \cite{ghaffari2022deterministic}, and $\widetilde{O}(\log^{5/3} n)$ rounds \cite{ghaffari2024near}.

Due to a reduction proved in \cite{jakob2025towards}, the aforementioned upper bounds for $(2\Delta - 1)$-edge coloring also apply to $(2\Delta - 2)$-edge coloring, with the exception of the first bound, which instead degrades slightly to $O(\log \log n + \log^{12} \Delta)$ rounds.
On the lower bound side, however, there is a substantial difference between the two cases: for $2\Delta - 1$ colors nothing better than an $\Omega(\log^* n)$-round lower bound is known \cite{linial1992locality} whereas the $\Omega(\log n)$-round lower bound discussed before applies as soon as the number of colors drops to $2\Delta - 2$ \cite{chang2019distributed}.
While no higher lower bound than $\Omega(\log n)$ rounds is known down to (and including) $\Delta + 1$ colors, the state-of-the-art upper bounds become progressively worse with the number of colors decreasing.
In particular, the best currently known bounds are $\widetilde{O}(\eps^{-2} \log^2 \Delta \log n)$ rounds for $((3/2 + \eps)\Delta)$-edge coloring (for any $\eps > 1/\Delta$) \cite{brandt2025locality}, $O(\Delta^2 \log n)$ rounds for $(3\Delta/2)$-edge coloring \cite{brandt2025locality}, and $O(\log^2 n) + \widetilde{O}(\log^2 \Delta \log n)$ rounds for $((1 + \eps)\Delta + O(\sqrt{\log n}))$-edge coloring (for any constant $\eps > 0$) \cite{blikstad2026deterministic}.

\paragraph{Existential results for bipartite edge coloring.}
While our results use $\Delta + 1$-colors, it should be noted that existentially, every bipartite graph can be edge-colored with $\Delta$ colors \cite{Konig1916}; however, computing such a coloring in the LOCAL model takes $\Omega(n)$ rounds already on (even) cycles.
In fact, as shown in \cite{galvin1995list}, every bipartite graph can even be list-edge-colored with $\Delta$ colors, i.e., if each edge is assigned a list of $\Delta$ colors (from some arbitrary color space), then it is possible to select, for each edge, a color from its list such that the resulting coloring is proper.
This result was strengthened in \cite{borodin1997list}: such a list edge coloring exists already if the list given to each edge $\{ u, v \}$ is of size $\max(\deg(u), \deg(v))$, which shows that also for the local-palette version of Vizing's theorem from \cite{christiansen2023power} discussed before, the assumption of bipartiteness reduces the number of required colors by $1$.

\subsection{Open Problems}\label{sec:open-problems}

The most immediate open problem is to remove the bipartiteness assumption in \Cref{thm:edge-col-main} (or even \Cref{thm:stronger-version}), i.e., to fully answer \Cref{que:logarithmic}.

A second, orthogonal open problem is to reduce the dependency on $\Delta$ as much as possible.
A first goal could be to obtain a polynomial bound on $\Delta$ (while still preserving the logarithmic dependency on $n$) but as currently no lower bound on the $\Delta$-dependency is known if a logarithmic dependency on $n$ is allowed, a perhaps even more important question is whether a clean $O(\log n)$-round algorithm (without any additional dependency on $\Delta$) exists.

As already alluded to before, from the perspective of developing new lower bound techniques (or extending existing ones), a further important open problem is to identify natural (non-global) locally checkable problems that may exhibit superlogarithmic lower bounds already on constant-degree graphs (and prove such superlogarithmic lower bounds).
Alternatively, it would be highly interesting to prove for wide classes of natural locally checkable problems exhibiting $\Omega(\log n)$-round lower bounds that these lower bounds are tight by proving matching upper bounds.

\subsection{Roadmap}
The remainder of the paper is split into three sections.
In \Cref{sec:prelims}, we recall relevant definitions and results from the literature and introduce new concepts central to our approach.
In \Cref{sec:density_hso}, we develop our extended algorithm for computing saturating matchings.
In \Cref{sec:edge-col}, we present our edge coloring algorithm and prove \Cref{thm:edge-col-main,thm:stronger-version}.

\section{Preliminaries}\label{sec:prelims}
In this section we fix notation and collect the definitions and 
basic observations used throughout the paper.

\paragraph{Graphs.}
For any graph $G$, let $V(G)$ and $E(G)$ denote the set of vertices and the set of edges of $G$, respectively. Let $n$ denote the number of nodes of $G$ and $\Delta$ denote its maximum degree. For two graphs $G_1=(V_1,E_1)$ and $G_2=(V_2,E_2)$, let $G_1\cup G_2$ be the graph $(V_1\cup V_2, E_1\cup E_2)$. For any subset of edges $E' \subseteq E(G)$, we denote by $V(E')$ the set of nodes in the subgraph induced by the edges in $E'$. If $E'=\{e\}$ for an edge $e$, we may simply denote it as $V(e)$.

\paragraph{Neighborhoods and distances.}
If $x$ and $y$ are neighbors in a graph $G$, we will call $y$ a \emph{$G$-neighbor} of $x$. We write the distance between nodes $u$ and $v$ in $G$ by $d_G(u,v)$. For $v \in V(G)$ we define $N_G(v) := \{ u \in V(G) \mid d_G(u,v) =1 \}$, the \emph{open neighborhood} of $v$. For $v \in V(G)$ and integer $t \geq 0$, let $B_G(v, t) := \{u \in V(G) : d_G(v, u) \leq t\}$ be the \emph{ball} of radius $t$ around $v$ in $G$.

\paragraph{Eccentricity and radius.}
For a graph $G$ and node $x \in V(G)$, the \emph{eccentricity of $x$ in $G$}, denoted \ecc{G}{x}, is the maximum distance from $x$ to any other node in $G$. The \emph{radius} of $G$, denoted \rad{G}, is the minimum eccentricity among all vertices of $G$. Note that the eccentricity of any node in a connected graph $G'$ is at most twice \rad{G'}.

\paragraph{The LOCAL model.}
In the \emph{LOCAL} model of distributed computing \cite{linial1992locality}, the input graph $G = (V, E)$ is the same as the communication network. Each vertex $v \in V$ is a processor that can perform arbitrarily large computation and has a unique $O(\log n)$-bit identifier, where $n := |V|$. Initially, each vertex knows only its own identifier, its degree, and the identifiers of its neighbors (together with any problem-specific input on itself and its incident edges). Computation proceeds in synchronous rounds: in each round, every vertex sends a message of arbitrary size to each of its neighbors, receives the messages sent to it, and updates its local state. After a finite number of rounds, every vertex must terminate and 
irrevocably output its own part of the solution, where different 
vertices may terminate in different rounds; these local outputs must 
together constitute a valid global solution to the problem. For 
example, if the problem is to find a maximal independent set of the 
input graph, each vertex outputs whether or not it belongs to the 
set, and the vertices that declare membership must indeed form a 
maximal independent set. The \emph{round complexity} (or \emph{time 
complexity}) of an algorithm is the number of rounds until all the 
vertices has terminated, in the worst case over all inputs (and 
identifier assignments).

This model admits an equivalent formulation: a deterministic 
algorithm with round complexity $T$ is precisely a function mapping 
the labeled radius-$T$ neighborhood of a vertex to that vertex's 
output, since in $T$ rounds information can travel at most $T$ hops, 
and conversely every vertex can collect its entire radius-$T$ 
neighborhood in $T$ rounds and then determine its output by local 
computation. An algorithm is \emph{randomized} if each vertex additionally has access to a private source of random bits.

\paragraph{Bipartite graphs.}
Let $G=(X\sqcup Y, E)$ be a bipartite graph. If we want to distinguish between nodes in $X$ and $Y$, we write $G$ as the triple $(X,Y,E)$; in this sense, $(X,Y,E)$ and $(Y,X,E)$ are considered different bipartite graphs. For a subgraph $H=(X',Y',E')$ of $G$ with $X'\subseteq X$ and $Y'\subseteq Y$, we may denote $X'$ and $Y'$ as $X(H)$ and $Y(H)$ respectively. We may refer to nodes in $X(H)$ as \emph{vertices on the $X$-side} or simply \emph{$X$-vertices}; the analogous convention applies to nodes in $Y(H)$.

For a bipartite graph $G=(X,Y,E)$ and subsets $X' \subseteq X$, $Y' \subseteq Y$, we denote by $G[X',Y']$ the bipartite graph $(X',Y',E')$ where $E'$ is the set of edges of $G$ between $X'$ and $Y'$.

A matching $M$ in $G$ is \emph{$X$-saturating} if every node in $X$ is an endpoint of an edge in $M$. Note that an $X$-saturating matching is a maximum matching in $G$.

\medskip
In the following, we formally introduce objects central to our overall approach (see \Cref{sec:approach}).
We start with a classification of vertices into good, bad, and neutral vertices.
\begin{definition}\label{def:good_bad_neutral}
    Let $k \geq 1$ be an integer and let $G = (X, Y, E)$ be a bipartite 
    graph such that every $x \in X$ has $\degree{G}{x} \geq k - 1$ and 
    every $y \in Y$ has $\degree{G}{y} \leq k$. We define the following 
    terminology for the vertices of $G$.
    \begin{itemize}
        \item A vertex $u \in V(G)$ is \emph{$k$-good in $G$} if either 
            $u \in X$ with $\degree{G}{u} \geq k$, or $u \in Y$ with 
            $\degree{G}{u} \leq k - 2$. Let $\goodnodes{V}(G,k)$ denote 
            the set of $k$-good vertices in $G$.
        \item A vertex $u \in V(G)$ is \emph{$k$-bad in $G$} if $u \in Y$ 
            and $\degree{G}{u} = k$. Let $\badnodes{V}(G,k)$ denote the 
            set of $k$-bad vertices in $G$.
        \item A vertex $u \in V(G)$ is \emph{$k$-neutral in $G$} if 
            $\degree{G}{u} = k - 1$. Let $\neutralnodes{V}(G,k)$ denote 
            the set of $k$-neutral vertices in $G$.
    \end{itemize}
    The degree assumption ensures that every vertex of $G$ is either 
    $k$-good, $k$-bad, or $k$-neutral in $G$.
\end{definition}

Next, we introduce the notion of a \emph{\net{k}{\alpha}{\beta}}, which will be central for our extended algorithm for computing a saturating matching. Informally, a \net{k}{\alpha}{\beta} is a bipartite graph where degrees are tightly controlled and $k$-bad vertices are both rare and well-separated. The formal definition follows.
\begin{definition}\label{def:net}
    Let $k \geq 1$ be an integer, and let $\alpha, \beta$ be positive 
    integers with $\beta \geq 2$. A bipartite graph $G = (X, Y, E)$ is 
    a \emph{\net{k}{\alpha}{\beta}} if it satisfies the following 
    three properties.
    \begin{enumerate}
        \item \emph{(Degree.)} Every $x \in X$ has $\degree{G}{x} \geq k - 1$, and every 
            $y \in Y$ has $\degree{G}{y} \leq k$. 
            \label[property]{property:net_degree}
        \item \emph{(Reach.)} For every $y \in Y$ there exists $u \in V(G)$ that is 
            $k$-good in $G$ with $d_G(y, u) \leq \alpha$. 
            \label[property]{property:net_reach}
        \item \emph{(Distance.)} Any two distinct vertices $y_1, y_2 \in Y$ that are 
            $k$-bad in $G$ satisfy $d_G(y_1, y_2) \geq \beta \alpha$. 
            \label[property]{property:net_distance}
    \end{enumerate}
\end{definition}
The following weaker variant relaxes the reach property to apply only at $k$-bad vertices.
\begin{definition}\label{def:weaknet}
    Let $k \geq 1$ be an integer, and let $\alpha, \beta$ be positive 
    integers with $\beta \geq 2$. A bipartite graph $G = (X, Y, E)$ is 
    a \emph{\weaknet{k}{\alpha}{\beta}} if it satisfies the following 
    three properties.
    \begin{enumerate}
        \item \emph{(Degree.)} Every $x \in X$ has $\degree{G}{x} \geq k - 1$, and every 
            $y \in Y$ has $\degree{G}{y} \leq k$. 
            \label[property]{property:weaknet_degree}
        \item \emph{(Reach.)} For every $y \in Y$ that is $k$-bad in $G$, there exists 
            $u \in V(G)$ that is $k$-good in $G$ with 
            $d_G(y, u) \leq \alpha$. 
            \label[property]{property:weaknet_reach}
        \item \emph{(Distance.)} Any two distinct vertices $y_1, y_2 \in Y$ that are 
            $k$-bad in $G$ satisfy $d_G(y_1, y_2) \geq \beta \alpha$. 
            \label[property]{property:weaknet_distance}
    \end{enumerate}
\end{definition}

These two notions are related in the obvious way.
\begin{observation}\label{obsvn:net_is_weaknet}
    Every \net{k}{\alpha}{\beta} is a \weaknet{k}{\alpha}{\beta}.
\end{observation}

Both of these notions also pass to connected components.
\begin{observation}\label{obsvn:component_of_net_is_net}
    Let $G = (X, Y, E)$ be a bipartite graph and let 
    $C = (X', Y', E')$ be a connected component of $G$ with 
    $X' \subseteq X$, $Y' \subseteq Y$, and $E' \subseteq E$. If $G$ 
    is a \net{k}{\alpha}{\beta}, then $C$ is a \net{k}{\alpha}{\beta}. 
    If $G$ is a \weaknet{k}{\alpha}{\beta}, then $C$ is a 
    \weaknet{k}{\alpha}{\beta}.
\end{observation}

\begin{proof}
  Every vertex of $C$ has all its neighbors in $G$ inside $C$, so 
$\degree{C}{v} = \degree{G}{v}$ for every $v \in V(C)$ and 
$d_C(u, v) = d_G(u, v)$ for every $u, v \in V(C)$. In particular, a 
vertex of $C$ is $k$-good, $k$-bad, or $k$-neutral in $C$ exactly 
when it is so in $G$.

Suppose $G$ is a \net{k}{\alpha}{\beta}. Then 
\Cref{property:net_degree} for $G$ directly gives 
\Cref{property:net_degree} for $C$, and 
\Cref{property:net_distance} for $G$ directly gives 
\Cref{property:net_distance} for $C$. For 
\Cref{property:net_reach}, given $y \in Y'$ there exists 
$u \in V(G)$ that is $k$-good in $G$ with $d_G(y, u) \leq \alpha$; 
since $d_G(y, u)$ is finite, $u \in V(C)$, and $u$ is $k$-good 
in $C$.

The argument for the weak-net case is identical.  
\end{proof}

We now focus on the notion of \emph{Hall graphs}, a local object used in our matching construction. This notion was introduced in \cite{brandt2025locality} in the setting of hypergraphs. We rephrase it to fit our setting.
\begin{definition} \label{def:Hall_graph}
    Let $G=(X,Y,E)$ be a bipartite graph. Let $H=(X',Y',E')$ be a subgraph of $G$ where $X'\subseteq X$, $Y'\subseteq Y$ and $E'\subseteq E$. We say that $H$ is a \emph{Hall graph in $G$} if 
    \begin{itemize}
        \item there exists an $X'$-saturating matching in $H$, and
        \item $N_G(Y')\subseteq X'$.
    \end{itemize}
\end{definition}

For the rest of this section, we record some definitions and properties that are related to Hall graphs. All of these properties have appeared in equivalent forms in \cite{brandt2025locality}.
The next lemma is a straightforward consequence of Hall's theorem: any bipartite graph $G=(X,Y,E)$ with no isolated vertices and $|X(G)| \leq |Y(G)|$ contains a non-empty Hall graph.   
\begin{lemma} \label{lem:hall_graph_exists_when_|X|_less_than_|Y|}
    Let $G=(X,Y,E)$ be a bipartite graph with no isolated nodes such that $|X|\leq |Y|$. Then there exists a non empty Hall graph $H$ in $G$.
\end{lemma}

\begin{proof}
    We construct induced subgraphs $G = G_0 \supseteq G_1 \supseteq \cdots$ 
    by repeatedly deleting specific non-empty set together with its 
    neighborhood, halting once Hall's condition holds. Set $X_0 := X$, 
    $Y_0 := Y$, $G_0 := G$. Given $G_i = G[X_i, Y_i]$, halt and output 
    $H := G_i$ if $|N_{G_i}(S)| \geq |S|$ for every $S \subseteq X_i$; 
    otherwise fix $S \subseteq X_i$ with $|N_{G_i}(S)| < |S|$ and set 
    $X_{i+1} := X_i \setminus S$, $Y_{i+1} := Y_i \setminus N_{G_i}(S)$, 
    $G_{i+1} := G[X_{i+1}, Y_{i+1}]$.

    We claim that for every $i$ for which $G_i$ is defined,
    \begin{align}
        Y_i &= N_{G_i}(X_i), \label{eq:hall_1}\\
        N_G(Y_i) &\subseteq X_i. \label{eq:hall_2}
    \end{align}
    Both hold at $i=0$, the former since $G$ has no isolated nodes and 
    the latter since all edges of $G$ run between $X$ and $Y$. Assume 
    both at $i$ and consider any $y \in Y_{i+1}$. Vertex $y$ lies in $Y_i \setminus N_{G_i}(S)$, 
    so by \cref{eq:hall_1} at $i$, it has a neighbor in $X_i \setminus S = 
    X_{i+1}$, giving \cref{eq:hall_1} at $i+1$; and any 
    $x' \in N_G(y)$ lies in $X_i$ by \cref{eq:hall_2} at $i$ and not in 
    $S$ (else $y \in N_{G_i}(S)$), giving \cref{eq:hall_2} at $i+1$.

    Induction on $i$ now gives $|X_i| \leq |Y_i|$: at the inductive step, 
    $|Y_{i+1}| = |Y_i| - |N_{G_i}(S)| \geq |Y_i| - |S| + 1 > 
    |X_i| - |S| = |X_{i+1}|$. Hence $S \subsetneq X_i$ at every step, 
    since $S = X_i$ would give $|N_{G_i}(X_i)| = |Y_i| \geq |X_i|$ by 
    \cref{eq:hall_1}, contradicting deficiency. So $X_{i+1} \neq 
    \emptyset$ throughout, and as $|X_i|$ strictly decreases the 
    construction halts at some $G_m$ with $X_m \neq \emptyset$. By 
    halting and Hall's theorem \cite{hall1987representatives}, $G_m$ 
    admits an $X_m$-saturating matching; together with 
    \cref{eq:hall_2}, $H := G_m$ is a non-empty Hall graph in $G$.
\end{proof}

The next definition formalizes what remains of a bipartite graph after a Hall graph is removed.
\begin{definition}
Let $G = (X, Y, E)$ be a bipartite graph and let $H = (X', Y', E')$ be a (possibly empty) Hall graph in $G$. We define 
$$ G \setminus H := G[\, X \setminus X',\; N_G(X \setminus X') \,].
$$
More generally, given bipartite graphs $H_1, H_2, \dots, H_t$ (which are possibly empty) 
such that $H_1$ is a Hall graph in $G$ and, for each 
$2 \leq i \leq t$, the graph $H_i$ is a Hall graph in 
$G \setminus \bigcup_{j < i} H_j$, we define recursively
$$
        G \setminus \bigcup_{i=1}^{t} H_i 
        := \left( G \setminus \bigcup_{i=1}^{t-1} H_i \right) \setminus H_t.
$$
We may refer to $ G \setminus \bigcup_{i=1}^{t} H_i$ as a \emph{residual graph} of $G$.
\end{definition}

We record some basic properties of residual graphs which will be useful later.
\begin{observation}\label{obsvn:residual_graph_properties}
    Let $G = (X, Y, E)$ be a bipartite graph. Given  bipartite graphs $H_1, H_2, \dots, H_t$ 
such that $H_1$ is a Hall graph in $G_0:=
G$ and, for each 
$2 \leq i \leq t$, the graph $H_i$ is a Hall graph in 
$G_{i-1}:=G \setminus \bigcup_{j < i} H_j$, the following holds by construction.
\begin{itemize}
    \item For $1\leq j \leq t$, $X(G_j)= X\setminus \bigcup_{i=1}^{j} X(H_i)$.
    \item For $1\leq j \leq t$, $N_G(X(G_j))=Y(G_j)$.
    \item For $1\leq i<j \leq t$, $X(G_j)\subseteq X(G_i)$.
\end{itemize}
\end{observation}

The following lemma shows that residual graphs interact well with Hall graphs: a union of Hall graphs taken in successive residual graphs is itself a Hall graph in the original graph.
\begin{lemma}\label{lem:union_of_hall_graphs_is_hall}
    Let $G = (X, Y, E)$ be a bipartite graph. Let 
    $H_1, H_2, \ldots, H_t$ be bipartite graphs such that $H_1$ is a 
    Hall graph in $G$ and, for each $2 \leq j \leq t$, the graph 
    $H_j$ is a Hall graph in $G \setminus \bigcup_{i < j} H_i$. Then 
    $\bigcup_{j=1}^{t} H_j$ is a Hall graph in $G$.
\end{lemma}

\begin{proof}
    Set $H := \bigcup_{j=1}^{t} H_j$,  $G_0 := G$, and 
    $G_j := G_{j-1} \setminus H_j$ for $1 \leq j \leq t$. An 
    induction on $j$ shows $$G_j = G[\,X \setminus \bigcup_{i \leq j} X(H_i),\  
    N_G(X \setminus \bigcup_{i \leq j} X(H_i))\ ],$$ 
    so every 
    $x \in X(G_j)$ has all of its $G$-neighbors inside $V(G_j)$.
    In particular, if $i<j$, then $X(G_{j-1}) \subseteq X(G_i)$, and hence
    $$
    Y(G_{j-1}) = N_G(X(G_{j-1})) \subseteq N_G(X(G_i)) = Y(G_i).
    $$

    We first show that $H$ admits an $X(H)$-saturating matching. For 
    each $j$, fix an $X(H_j)$-saturating matching $M_j$ of $H_j$, 
    which exists since $H_j$ is a Hall graph in $G_{j-1}$. We have that 
    $X(H_1), \ldots, X(H_t)$ are pairwise disjoint, because 
    $X(H_j) \subseteq X(G_{j-1})$ and 
    $X(G_{j-1}) \cap X(H_i) = \emptyset$ for every $i < j$ by the 
    definition of $G_{j-1}$. Hence $M := \bigsqcup_{j=1}^{t} M_j$ is 
    a set of edges where for each edge the $X$-endpoint of the edge in $M$ lies in 
    exactly one $X(H_j)$, and the $Y$-endpoint of the edge lies in $Y(H_j)$ for 
    the same $j$. To see that $M$ is a matching, suppose $e \in M_i$ 
    and $e' \in M_j$ with $i < j$ share an endpoint. That endpoint 
    cannot lie in $X$, since $X(H_i) \cap X(H_j) = \emptyset$; so it 
    lies in $Y$, and is then a $Y$-endpoint of $e' \in M_j \subseteq E(H_j)$, 
    hence in $Y(H_j) \subseteq Y(G_{j-1})$. On the other hand, the same vertex is the $Y$-endpoint of 
$e \in M_i \subseteq E(H_i)$, so $y \in Y(H_i)$. Since $H_i$ is a 
Hall graph in $G_{i-1}$, we have
$$
    N_{G_{i-1}}(y) \subseteq X(H_i).
$$
Hence $y$ has no neighbor in $X(G_i)$, because
$X(G_i) \cap X(H_i) = \emptyset$, and therefore $y \notin Y(G_i)$.
But $y \in Y(H_j) \subseteq Y(G_{j-1}) \subseteq Y(G_i)$, a contradiction.
Hence $M$ is a matching, and it saturates $X(H)$ since each $M_j$
saturates $X(H_j)$.
    
    We next show that $N_G(Y(H)) \subseteq X(H)$. Let $y \in Y(H)$, 
    so $y \in Y(H_j)$ for some $j$, and let $w \in N_G(y)$. Since 
    $H_j$ is a Hall graph in $G_{j-1}$, $N_{G_{j-1}}(y) \subseteq X(H_j)$. 
    If $w \in X(G_{j-1})$, then $w$ is a $G_{j-1}$-neighbor of $y$, 
    so $w \in X(H_j) \subseteq X(H)$. Otherwise $w \in X \setminus 
    X(G_{j-1}) = \bigcup_{i < j} X(H_i) \subseteq X(H)$. Either way, 
    $w \in X(H)$.

    Hence $H$ is a Hall graph in $G$.
\end{proof}

\section{Saturating Matching in Nets} \label{sec:density_hso}

Fix a \net{k}{\alpha}{\beta} $G = (X, Y, E)$ on $n\geq 2$ nodes; we assume this throughout the section unless stated otherwise. Our aim is \Cref{thm:Density-HSO}: a fast deterministic LOCAL algorithm that computes an $X$-saturating matching in $G$.

The strategy we use mirrors the one in \cite{brandt2025locality} but it is generalized to our setting. Each vertex $x \in X$ computes, using only its local neighborhood in $G$, a Hall graph $H_x$ in $G$ that contains $x$ on its $X$-side, together with an $X(H_x)$-saturating matching of $H_x$. A tie-breaking rule across $X$ retains a subcollection of these Hall graphs such that every $x \in X$ still lies on the $X$-side of some retained Hall graph, and the saturating matchings inside the retained Hall graphs together form an $X$-saturating matching of $G$.

Two ingredients drive this strategy. The first is a local existence result: for every $x \in X$, a Hall 
graph in $G$ containing $x$ on its $X$-side must exist within a small ball around $x$. This is the content of 
\Cref{lem:hall_graph_containing_node_in_net}, the central 
lemma of the section. Its proof iterates the following step: 
either the component of the current graph containing $x$ is 
small enough to serve itself as a Hall graph containing $x$, 
or a density based argument extracts a Hall graph from a small ball 
around $x$ (which, however, need not contain $x$ on its 
$X$-side); in the latter case we remove the extracted Hall 
graph and repeat on the residual graph. Once this process terminates (which it will), the successively 
removed Hall graphs are glued back together 
by \Cref{lem:union_of_hall_graphs_is_hall} into a single Hall 
graph in $G$ containing $x$.

At the heart of the iteration is 
\Cref{lem:slow-growing-ball}: in a net whose component around 
$x$ has large radius, some ball around $x$ grows slowly, and 
the induced bipartite subgraph on such a ball contains a Hall 
graph (via 
\Cref{lem:hall_graph_exists_when_|X|_less_than_|Y|}). For 
this argument to be applied at every step of the iteration, 
the residual graphs must retain enough structure of a net. 
This is what the structural lemmas \Cref{lem:residual-reach} 
and \Cref{lem:residual-is-weaknet-or-net} deliver: components 
of residual graphs are weak-nets, and are either of small 
radius --- in which case they admit saturating matchings by 
\Cref{lem:weaknet_hall} and can serve as Hall graphs 
containing $x$ --- or are again nets with slightly weaker 
parameters, on which \Cref{lem:slow-growing-ball} can be reapplied.

The second ingredient is the way the local Hall graphs are 
used to make a global matching, carried out in the proof of 
\Cref{thm:Density-HSO}. Once each $x \in X$ has selected a 
Hall graph $H_x$ containing it, a consistent tie-breaking 
rule based on vertex identifiers determines which $H_x$'s are 
retained. 

\subsection{Saturating Matching in Weak Nets}
We now begin with the following matching result. It shows that any \weaknet{k}{\alpha}{\beta} admits an $X$-saturating matching.
\begin{lemma}\label{lem:weaknet_hall}
    Let $k \geq 2$ be an integer and let $\alpha, \beta$ be positive 
    integers with $\alpha \geq 3$ and $\beta \geq 3$. If 
    $G = (X, Y, E)$ is a \weaknet{k}{\alpha}{\beta}, then $G$ admits an $X$-saturating matching.
\end{lemma}

\begin{proof}
     By Hall's theorem\cite{hall1987representatives}, it suffices to show that $|N_G(S)| \geq |S|$ for 
    every $S \subseteq X$. Suppose for a contradiction that there exists 
    $S \subseteq X$ with $|N_G(S)| \leq |S| - 1$. We may assume that 
    $G' := G[S, N_G(S)]$ is connected, for if the components of $G'$ 
    partition $S$ as $S_1 \sqcup \cdots \sqcup S_t$, then 
    $\sum_i |N_G(S_i)| = |N_G(S)| \leq |S| - 1 = \sum_i |S_i| - 1$, so 
    $|N_G(S_i)| \leq |S_i| - 1$ for some $i$, and we may replace $S$ 
    by $S_i$.

    Let $b$ be the number of $k$-bad vertices in $G'$, and let $g_1$ 
    and $g_2$ be the numbers of $k$-good vertices in $G'$ lying in 
    $S$ and $N_G(S)$, respectively. We derive a contradiction by 
    showing that $g_1 + g_2 < b$ and $g_1 + g_2 \geq b$.

    Note that every neighbor of $x \in S$ lies in $N_G(S)$, so by 
    \Cref{property:weaknet_degree} the edge count of $G'$ satisfies
    $$
        |E(G')| \geq (k-1)(|S| - g_1) + k g_1 = (k-1)|S| + g_1.
    $$
    On the other hand, again by \Cref{property:weaknet_degree},
    $$
        |E(G')| \leq (k-1)(|N_G(S)| - b - g_2) + k b + (k-2) g_2 
                = (k-1)|N_G(S)| + b - g_2.
    $$
    Combining the above two inequalities with $|N_G(S)| \leq |S| - 1$ yields 
    $g_1 + g_2 \leq b - (k - 1) < b$, since $k \geq 2$.

    To show $g_1 + g_2 \geq b$, we show an injection from the 
    $k$-bad vertices of $G'$ into the $k$-good vertices of $G'$. Let 
    $v$ be $k$-bad in $G'$ (which implies $v\in Y(G')$ as well). Then 
    $k = \degree{G'}{v} \leq \degree{G}{v} \leq k$, so $v$ is also 
    $k$-bad in $G$, and by \Cref{property:weaknet_reach} there exists 
    a $k$-good node in $G$ within distance $\alpha$ in $G$ from $v$. Fix a path $P_v$ from $v$ to the closest such $k$-good node, say
    $w(v)$, in $G$, and let $z(v)$ denote the last vertex of $P_v$ 
    that lies in $V(G')$. We claim that $v \mapsto z(v)$ is a 
    well-defined injection from $\badnodes{V}(G',k)$ into $\goodnodes{V}(G',k)$.

    Note first that $z(v) \neq v$.  Suppose for contradiction   that $z(v) = v$. Since $v$ is $k$-bad in $G$, 
    it is not $k$-good in $G$, so $w(v) \neq v$ and $P_v$ has 
    length at least one. Then the successor of $v$ along $P_v$ lies outside 
    $V(G')$; since $v \in N_G(S) \subseteq Y$, all neighbors of $v$ 
    lie in $X$, so this successor is in $X \setminus S$. Hence 
    $\degree{G'}{v} < \degree{G}{v} = k$, contradicting that $v$ 
    is $k$-bad in $G'$.

    Note next that every internal vertex $u$ of $P_v$ is $k$-neutral 
    in $G$. Indeed, $u$ is not $k$-good in $G$ by minimality of 
    $d_G(v, w(v))$, and $u$ is not $k$-bad in $G$ either, since $u\neq v$ and
    $d_G(v, u) \leq \alpha - 1 < \beta\alpha$ would violate 
    \Cref{property:weaknet_distance}.

    We now show that $z(v)$ is $k$-good in $G'$. Suppose first that 
    $z(v) = w(v)$. Then $z(v)$ is $k$-good in $G$, and this status is 
    preserved in $G'$: if $z(v) \in X$ then all its neighbors remain 
    in $G'$ and $\degree{G'}{z(v)} = \degree{G}{z(v)} \geq k$; if 
    $z(v) \in Y$ then 
    $\degree{G'}{z(v)} \leq \degree{G}{z(v)} \leq k - 2$. Now suppose 
    that $z(v) \neq w(v)$. Then $z(v)$ is internal to $P_v$, hence 
    $k$-neutral in $G$, and the successor of $z(v)$ along $P_v$ lies 
    outside $V(G')$. Consequently $z(v) \notin S$, for otherwise every 
    neighbor of $z(v)$ would lie in $N_G(S) \subseteq V(G')$. Therefore 
    $z(v) \in N_G(S) \subseteq Y$ and has at least one neighbor outside 
    $V(G')$, giving 
    $\degree{G'}{z(v)} \leq \degree{G}{z(v)} - 1 = k - 2$.

    Finally, we verify injectivity. Suppose $v_1 \neq v_2$ are both 
    $k$-bad in $G'$ with $z(v_1) = z(v_2) = z$. The subpath of 
    $P_{v_i}$ from $v_i$ to $z$ lies entirely in $V(G')$, so 
    $d_G(v_i, z) \leq \alpha$, and hence 
    $d_G(v_1, v_2) \leq 2\alpha$. On the other hand, since both 
    $v_1$ and $v_2$ are $k$-bad in $G$, 
    \Cref{property:weaknet_distance} gives 
    $d_G(v_1, v_2) \geq \beta\alpha \geq 3\alpha$,
    contradicting the previous bound.

    Combining the two parts, $b \leq g_1 + g_2 < b$, a contradiction. 
    Therefore $|N_G(S)| \geq |S|$ for every $S \subseteq X$, and by 
    Hall's theorem $G$ admits an $X$-saturating matching.
\end{proof}

\subsection{Structural Properties of Residual Graphs}

When we remove a Hall graph (or multiple ones) from a \net{k}{\alpha}{\beta}, the residual graph remaining need no longer satisfy \Cref{property:net_reach}. The next lemma says that within each component of this residual graph, every $Y$-vertex either lies in a small-radius component, or still has a nearby $k$-good vertex ---albeit at the cost of this good node now only being within distance $3\alpha$ instead of $\alpha$.

\begin{lemma}\label{lem:residual-reach}
    Let $k \geq 1$ be an integer and let $\alpha, \beta$ be positive 
    integers with $\alpha \geq 3$ and $\beta \geq 3$. Let 
    $G = (X, Y, E)$ be a \net{k}{\alpha}{\beta}, and set $G_0 := G$. 
    For $1 \leq i \leq j$, let $H_i$ be a Hall graph (possibly empty) 
    in $G_{i-1}$ and let $G_i := G_{i-1} \setminus H_i$. Let $C$ be a 
    connected component of $G_i$ for some $0 \leq i \leq j$, and let 
    $w \in Y\cap V(C)$. Then at least one of the following holds:
    \begin{itemize}
        \item \rad{C} is at most $3\alpha$;
        \item $w$ lies within distance $3\alpha$ in $G_i$ of a vertex 
            that is $k$-good in $G_i$.
    \end{itemize}
\end{lemma}

\begin{proof}
    All distances below are measured in $G_0$ unless otherwise stated. 
    Suppose \rad{C} exceeds $3\alpha$; we produce a vertex 
    that is $k$-good in $G_i$ within distance $3\alpha$ of $w$ in 
    $G_i$.

   A straightforward induction on $i$ shows that
$G_0[\, S_i,\, N_{G_0}(S_i) \,] \subseteq G_i$, where
$S_i := X \cap V(G_i)$. Indeed, the claim is immediate for $i = 0$,
and if it holds for some $i \geq 0$, then every $x \in S_i$ has all
of its $G_0$-neighbors, together with the incident edges, inside
$G_i$; hence $N_{G_i}(S') = N_{G_0}(S')$ for every
$S' \subseteq S_i$. Since
$G_{i+1} = G_i[\, S_{i+1},\, N_{G_i}(S_{i+1}) \,]$ by definition and
$S_{i+1} \subseteq S_i$, we get
$G_{i+1} \supseteq G_0[\, S_{i+1},\, N_{G_0}(S_{i+1}) \,]$, as
claimed. Consequently, every $x \in S_i$ has all of its neighbors in
$G_0$ inside $V(G_i)$, and so any path in $G_0$ that leaves
$V(G_i)$ must do so from a vertex in $Y$.

    By \Cref{property:net_reach} there exists $z \in V(G_0)$ that is 
    $k$-good in $G_0$ with $d(w, z) \leq \alpha$; fix a shortest path 
    $P_1'$ from $w$ to $z$ in $G_0$. The internal vertices of $P_1'$ 
    are not $k$-good in $G_0$ by minimality of $d(w, z)$, and by 
    \Cref{property:net_distance} at most one of them is $k$-bad in 
    $G_0$, since any two internal vertices of a path of length at most 
    $\alpha$ lie at distance at most $\alpha-2 < \beta\alpha$. Let $v$ 
    denote this bad vertex if it exists; the remaining internal 
    vertices of $P_1'$ are $k$-neutral in $G_0$.

    Let $P_1$ be the maximal subpath of $P_1'$ starting at $w$ whose 
    vertices all lie in $V(G_i)$, and let $u_1$ be its last vertex. 
    Then $P_1$ is a path in $G_i$, so $d_{G_i}(w, u_1) \leq |P_1|$.

    Suppose first that $u_1 \neq v$. If $u_1 = z$, then $z \in V(G_i)$, 
    and it is still $k$-good in $G_i$: if $z \in X$ then all its 
    neighbors remain in $G_i$ and $\degree{G_i}{z} = \degree{G_0}{z} 
    \geq k$; if $z \in Y$ then 
    $\degree{G_i}{z} \leq \degree{G_0}{z} \leq k - 2$. Otherwise $u_1$ 
    is internal to $P_1'$ and distinct from $v$, hence $k$-neutral in 
    $G_0$; since any path in $G_0$ that leaves $V(G_i)$ must do so 
    from a vertex in $Y$, we have $u_1 \in Y$, and since the successor 
    of $u_1$ on $P_1'$ lies outside $V(G_i)$, we have 
    $\degree{G_i}{u_1} \leq \degree{G_0}{u_1} - 1 = k - 2$. In either 
    case $u_1$ is $k$-good in $G_i$ and 
    $d_{G_i}(w, u_1) \leq |P_1'| \leq \alpha \leq 3\alpha$.

    Now suppose $u_1 = v$. If $v$ loses at least two neighbors in 
    passing from $G_0$ to $G_i$, then $\degree{G_i}{v} \leq k - 2$, 
    so $v$ is $k$-good in $G_i$ and we are done as before. We may 
    therefore assume $v$ loses exactly one neighbor, so 
    $\degree{G_i}{v} = k - 1$.

    Note that $|P_1| = d(w, v)$ is even, since $w \in Y$ and $v \in Y$. 
    Combined with $|P_1| \leq |P_1'| - 1 \leq \alpha - 1$ (as $v$ is 
    internal to $P_1'$), this yields
    $$
        |P_1| \leq 
        \begin{cases}
            \alpha - 2 & \text{if } \alpha \text{ is even},\\
            \alpha - 1 & \text{if } \alpha \text{ is odd}.
        \end{cases}
    $$

  Set
    $$
        \lambda :=
        \begin{cases}
            \alpha     & \text{if } \alpha \text{ is even},\\
            \alpha + 1 & \text{if } \alpha \text{ is odd}.
        \end{cases}
    $$
    Since \rad{C} exceeds $3\alpha$, \ecc{C}{v} is greater than $3\alpha$, 
    and hence greater than $\lambda$. Thus there is a shortest path $P_2$ 
    in $C$ of length exactly $\lambda$ starting at $v$. Let $v^*$ denote 
    its other endpoint; since $C$ is a connected component of $G_i$, we 
    have $d_{G_i}(v, v^*) = \lambda$. Since $\lambda$ is even and 
    $v \in Y$, we have $v^* \in Y$.

    Among all vertices that are $k$-good in $G_0$, choose $z'$ to 
    minimize $d(v^*, z')$; by \Cref{property:net_reach} applied to 
    $v^*$, $d(v^*, z') \leq \alpha$. Fix a shortest path $P_3'$ from 
    $v^*$ to $z'$ in $G_0$.

    We first claim that every internal vertex of $P_3'$ other than $v$ 
    is $k$-neutral in $G_0$. Indeed, such a vertex is not $k$-good in 
    $G_0$ by the minimality of $d(v^*, z')$. It is also not $k$-bad in 
    $G_0$: a $k$-bad vertex $b \neq v$ on $P_3'$ would satisfy
    $$
        d(b, v) \;\leq\; d(b, v^*) + d(v^*, v) 
                \;\leq\; |P_3'| + \lambda 
                \;\leq\; 2\alpha + 1 
                \;<\;    3\alpha 
                \;\leq\; \beta\alpha,
    $$
    using $\alpha \geq 3$ and $\beta \geq 3$, contradicting 
    \Cref{property:net_distance}, since both $b$ and $v$ are $k$-bad 
    in $G_0$.

    Let $P_3$ be the maximal subpath of $P_3'$ starting at $v^*$ whose 
    vertices all lie in $V(G_i)$, and let $u_3$ be its last vertex. 
    Note that any two consecutive vertices of $P_3'$ that both lie in 
    $V(G_i)$ are joined by an edge of $G_i$, since one of them lies in 
    $X \cap V(G_i)$ and every such vertex retains all of its 
    $G_0$-neighbors and incident edges in $G_i$; hence $P_3$ is a path 
    in $G_i$ and $d_{G_i}(v^*, u_3) \leq |P_3| \leq \alpha$.

    We next claim that $u_3 \neq v$. If $v$ lies on $P_3'$, then $v$ is internal to $P_3'$: indeed 
    $v \neq z'$, since $z'$ is $k$-good in $G_0$ while $v$ is $k$-bad 
    in $G_0$, and $v \neq v^*$, since $d_{G_i}(v, v^*) = \lambda > 0$. The 
    subpath of $P_3'$ from $v^*$ to $v$ then has length at most 
    $|P_3'| - 1 \leq \alpha - 1 < \lambda = d_{G_i}(v^*, v)$. Hence 
    this subpath is not a path in $G_i$, so it contains a vertex 
    outside $V(G_i)$, and therefore $P_3$ ends strictly before $v$. In 
    particular $u_3 \neq v$. On the other hand, if $v$ 
    does not lie on $P_3'$ at all, then trivially $u_3 \neq v$.

    We now show that $u_3$ is $k$-good in $G_i$. If $u_3 = z'$, then 
    $z'$ is $k$-good in $G_0$ which implies it is $k$-good in $G_i$ as well: 
    if $z' \in X$ then all its neighbors remain in $G_i$ and 
    $\degree{G_i}{z'} = \degree{G_0}{z'} \geq k$; if $z' \in Y$ then 
    $\degree{G_i}{z'} \leq \degree{G_0}{z'} \leq k - 2$. Otherwise 
    $u_3$ is internal to $P_3'$ and distinct from $v$, hence 
    $k$-neutral in $G_0$ by the first claim, and the successor of 
    $u_3$ on $P_3'$ lies outside $V(G_i)$. Since any path in $G_0$ 
    that leaves $V(G_i)$ must do so from a vertex in $Y$, we have 
    $u_3 \in Y$, and therefore 
    $\degree{G_i}{u_3} \leq \degree{G_0}{u_3} - 1 = k - 2$, so $u_3$ 
    is $k$-good in $G_i$.

    Finally, the concatenation $P_1 + P_2 + P_3$ is a walk in 
    $G_i$ from $w$ to $u_3$. Its length is bounded:
   
            $$|P_1| + |P_2| + |P_3| 
            \leq (\alpha - 1) + (\alpha + 1) + \alpha 
            = 3\alpha.$$

    Hence $d_{G_i}(w, u_3) \leq 3\alpha$, completing the proof.

\end{proof}

The next lemma examines the structure of residual graphs: every component of a residual graph of a \net{k}{\alpha}{\beta} is itself a \weaknet{k}{\alpha}{\beta}, and is either small (with radius at most $\beta\alpha/2$) or a \net{k}{3\alpha}{\lfloor \beta/3 \rfloor}.
\begin{lemma}\label{lem:residual-is-weaknet-or-net}
    Let $k \geq 1$ be an integer and let $\alpha, \beta$ be positive 
    integers with $\alpha \geq 3$ and $\beta \geq 60$. Let 
    $G = (X, Y, E)$ be a \net{k}{\alpha}{\beta}, and set $G_0 := G$. For 
    $1 \leq i \leq j$, let $H_i$ be a Hall graph (possibly empty) in 
    $G_{i-1}$ and let $G_i := G_{i-1} \setminus H_i$. Let $C$ be a 
    connected component of $G_i$ for some $0 \leq i \leq j$. Then 
   the following holds:
    \begin{enumerate}
        \item $C$ is a \weaknet{k}{\alpha}{\beta} and, \label[instance]{instance_1} 
        \item \rad{C} is at most $\frac{\beta\alpha}{2}$, or $C$ is a \net{k}{3\alpha}{\lfloor \beta/3 \rfloor}. \label[instance]{instance_2} 

    \end{enumerate}
\end{lemma}

\begin{proof}
    All distances below are measured in $G_0$ unless otherwise stated. 
A straightforward induction on $i$ shows that
$G_0[\, S_i,\, N_{G_0}(S_i) \,] \subseteq G_i$, where
$S_i := X \cap V(G_i)$. Indeed, the claim is immediate for $i = 0$,
and if it holds for some $i \geq 0$, then every $x \in S_i$ has all
of its $G_0$-neighbors, together with the incident edges, inside
$G_i$; hence $N_{G_i}(S') = N_{G_0}(S')$ for every
$S' \subseteq S_i$. Since
$G_{i+1} = G_i[\, S_{i+1},\, N_{G_i}(S_{i+1}) \,]$ by definition and
$S_{i+1} \subseteq S_i$, we get
$G_{i+1} \supseteq G_0[\, S_{i+1},\, N_{G_0}(S_{i+1}) \,]$, as
claimed. Consequently, every $x \in S_i$ has all of its neighbors in
$G_0$ inside $V(G_i)$, and so any path in $G_0$ that leaves
$V(G_i)$ must do so from a vertex in $Y$.

We first show that $C$ is a \weaknet{k}{\alpha}{\beta}. The 
bounds in \Cref{property:weaknet_degree} for $C$ are immediate from \Cref{property:net_degree} for $G_0$: any 
$x \in X \cap V(C)$ lies in $S_i$, so 
$\degree{C}{x} = \degree{G_0}{x} \geq k - 1$, and any 
$y \in Y \cap V(C)$ satisfies 
$\degree{C}{y} \leq \degree{G_0}{y} \leq k$.

For \Cref{property:weaknet_distance} for $C$, let 
$y_1, y_2 \in Y \cap V(C)$ be distinct vertices that are 
$k$-bad in $C$. Then for $q\in\{1,2\}$,
$\degree{G_0}{y_q} \geq \degree{C}{y_q} = k$ and 
$\degree{G_0}{y_q} \leq k$, so $\degree{G_0}{y_q} = k$; thus 
$y_1, y_2$ are $k$-bad in $G_0$ as well. Since $C$ is an 
induced subgraph of $G_0$, distances in $C$ dominate those in 
$G_0$, giving 
$d_C(y_1, y_2) \geq d_{G_0}(y_1, y_2) \geq \beta\alpha$ by 
\Cref{property:net_distance} for $G_0$.

For \Cref{property:weaknet_reach} for $C$, let $y \in Y \cap V(C)$ be 
$k$-bad in $C$. By \Cref{property:net_reach} for $G_0$, there exists a 
vertex that is $k$-good in $G_0$ within distance $\alpha$ of $y$. 
Among all vertices that are $k$-good in $G_0$, choose $z$ to 
minimize $d_{G_0}(y, z)$; then $d_{G_0}(y, z) \leq \alpha$. Fix a shortest 
path $P'$ from $y$ to $z$ in $G_0$. Every internal vertex of $P'$ 
is $k$-neutral in $G_0$: it cannot be $k$-good in $G_0$ since it 
lies strictly closer to $y$ than $z$ does, contradicting the 
minimality of $d_{G_0}(y,z)$; and it cannot be $k$-bad in $G_0$, since such a 
vertex $u$ would be distinct from $y$ and satisfy 
$d(y, u) \leq |P'| - 1 \leq \alpha - 1 < \beta\alpha$, 
contradicting \Cref{property:net_distance} for $G_0$.

Let $P$ be the maximal subpath of $P'$ starting at $y$ whose 
vertices all lie in $V(G_i)$, and let $v$ be its last vertex. 
Since $y$ is $k$-bad in $C$, its degree in $G_i$ equals its 
degree in $G_0$, so all of its neighbors in $G_0$ lie in 
$V(G_i)$; in particular the second vertex of $P'$ lies in 
$V(G_i)$, so $v \neq y$. If $v = z$, then $z \in V(G_i)$ and 
it is $k$-good in $G_i$ as it is so in $G_0$: a vertex 
in $X \cap V(G_i)$ retains all its neighbors in $G_0$, so 
$\degree{G_i}{z} = \degree{G_0}{z} \geq k$ when $z \in X$, 
and a vertex in $Y$ can only lose neighbors, giving 
$\degree{G_i}{z} \leq \degree{G_0}{z} \leq k - 2$ when 
$z \in Y$. Otherwise $v$ is internal to $P'$, hence $k$-neutral in $G_0$, and the successor of $v$ on $P'$ lies outside $V(G_i)$. Since any path in $G_0$ that leaves $V(G_i)$ must do so from a vertex in $Y$, this forces $v \in Y$, so $\degree{G_i}{v} \leq \degree{G_0}{v} - 1 = k - 2$. In either 
case $v$ is $k$-good in $G_i$, and 
$d_{G_i}(y, v) \leq |P| \leq \alpha$. Since degrees and distances inside $C$ 
match those in $G_i$, $v$ is $k$-good in $C$, and 
$d_{C}(y, v) \leq |P| \leq \alpha$ This verifies 
\Cref{property:weaknet_reach} for $C$, completing the proof that $C$ 
is a \weaknet{k}{\alpha}{\beta}. This gives that \Cref{instance_1} holds.

We now use \Cref{lem:residual-reach}, whose assumptions 
($\alpha \geq 3$ and $\beta \geq 3$) hold here. It gives that 
either \rad{C} is at most $3\alpha$, or every 
$w \in Y \cap V(C)$ lies within distance $3\alpha$ in $G_i$ 
of a vertex that is $k$-good in $G_i$.

Suppose first that \rad{C} is at most $3\alpha$. 
Since $\beta \geq 60$ gives 
$\frac{\beta\alpha}{2} \geq 30\alpha \geq 3\alpha$, \rad{C} is at most $\frac{\beta\alpha}{2}$.

Suppose instead that every $w \in Y \cap V(C)$ lies within 
distance $3\alpha$ in $G_i$ of a vertex that is $k$-good in 
$G_i$. Since degrees and distances inside a component of $G_i$ 
match those in $G_i$, this gives that every 
$w \in Y \cap V(C)$ lies within distance $3\alpha$ in $C$ of a 
vertex that is $k$-good in $C$, which is the condition of 
\Cref{property:net_reach} for $C$ as a 
\net{k}{3\alpha}{\lfloor \beta/3 \rfloor}. Combined with the 
degree bounds already established and with the fact that for any two k-bad vertices $y_1,y_2$ in $C$
$d_C(y_1, y_2) \geq \beta\alpha 
    \geq \lfloor \beta/3 \rfloor \cdot 3\alpha$ 
(using $3 \lfloor \beta/3 \rfloor \leq \beta$), $C$ is a 
\net{k}{3\alpha}{\lfloor \beta/3 \rfloor}.
This gives that \Cref{instance_2} holds.
\end{proof}

\subsection{Existence of Hall Graphs in Nets}

We now come to the argument at the core of this section. If the component of $G$ containing $x$ has large radius, then balls around $x$ cannot all grow geometrically, so at some distance $t$ the ball $B_G(x, t)$ grows by a factor of at most $1+\epsilon$ for an $\epsilon>0$ we choose . The net structure provides a lower bound on $k$-good vertices and an upper bound on $k$-bad vertices in this ball, and these bounds combine via an edge-counting argument to show that the induced bipartite subgraph on $B_G(x, t)$ has at least as many $Y$-vertices as $X$-vertices; \Cref{lem:hall_graph_exists_when_|X|_less_than_|Y|} then yields a Hall graph. 

The strategy behind the following lemma --- locating a slowly 
growing ball around $x$ in which the $Y$-side is at least as large as the $X$-side, and extracting a Hall graph from it via \Cref{lem:hall_graph_exists_when_|X|_less_than_|Y|} --- is due to \cite{brandt2025locality}, where a weaker analogue of the lemma is proved in a more restrictive setting where there are no $k$-bad and $k$-neutral nodes in $G$. Our version generalizes theirs to \net{k}{\alpha}{\beta}s; the presence of $k$-bad as well as $k$-neutral vertices requires a more delicate argument inside the ball.
\begin{lemma}\label{lem:slow-growing-ball}
    Let $G = (X, Y, E)$ be a \net{k}{\alpha}{\beta} on $n$ vertices 
    with $k \geq 2$, $\alpha \geq 3$, $\beta \geq 60$, and maximum 
    degree $\Delta \geq 2$. Let $x \in X$ and let $C$ be the connected 
    component of $G$ containing $x$. Let $\epsilon$ be a real number 
    with $0 < \epsilon \leq \frac{1}{3(k+1)\Delta^{\alpha+2}}$, let 
    $s := 4\alpha + 6$, let $L$ be an integer with 
    $L \geq \lceil \log_{1+\epsilon}(n) \rceil + 1$, and let $t_0$ be 
    an even integer with 
    $t_0 > \max\{\beta\alpha,\, 10\Delta^{\alpha+2}\}$. If 
    \rad{C} is at least $t_0 + s(L-1)$, then there exists an even 
    integer $t$ with $t_0 \leq t \leq t_0 + s(L-1)$ such that
    $$
        |B_G(x, t)| \;\leq\; (1+\epsilon)\cdot |B_G(x, t-s)|,
    $$
    and the induced bipartite subgraph 
    $G[\,X \cap B_G(x,t),\; Y \cap B_G(x,t)\,]$ contains a non-empty 
    Hall graph in $G$.
\end{lemma}
\begin{proof}
    We begin by finding an even integer $t$ in 
    $$
        T := \{t_0,\; t_0 + s,\; \ldots,\; t_0 + s(L-1)\}
    $$ 
    for which $|B_G(x,t)| \leq (1+\epsilon)\cdot|B_G(x,t-s)|$. Suppose 
    for contradiction that 
    $|B_G(x,t)| > (1+\epsilon)\cdot|B_G(x,t-s)|$ for every $t\in T$. 
    Applying this inequality along the $L$ elements of $T$ in succession,
    $$
        |B_G(x, t_0 + s(L-1))| 
            > (1+\epsilon)^L\cdot |B_G(x, t_0 - s)| 
            \geq (1+\epsilon)^L.
    $$
    Since $L \geq \log_{1+\epsilon}(n) + 1$, the right-hand side is at 
    least $(1+\epsilon)n > n$, contradicting 
    $|B_G(x, t_0 + s(L-1))| \leq n$. Hence some such $t$ exists, and by 
    construction $t$ is even and $t \geq t_0$. Fix it for the remainder of the proof.

    Set 
    $$
        X^* := X \cap B_G(x, t), \quad 
        Y^* := Y \cap B_G(x, t), \quad 
        G' := G[X^*, Y^*], \quad 
        B'' := B_G(x, t - s).
    $$
    By the choice of $t$, $|G'| \leq (1+\epsilon)\,|B''|$. Since $t$ is 
    even and $x \in X$, every $y \in Y^*$ satisfies $d_G(x,y) \leq t-1$, 
    so every $G$-neighbor of $y$ lies within distance at most $t$ from $x$ 
    and therefore in $X^*$; consequently 
    $$
        N_G(Y^*) \subseteq X^*.
    $$
    Moreover $G'$ has no isolated vertices: $x$ has at least $k-1$ 
    neighbors in $B_G(x,1) \subseteq V(G')$, and every other 
    $v \in V(G')$ has, as the predecessor on a shortest $x$--$v$ path, 
    a $G$-neighbor at distance less than $t$ from $x$, which lies in 
    $V(G')$.

    \medskip
    The remainder of the proof is devoted to showing that 
    $|X^*| \leq |Y^*|$. The argument proceeds in three parts. First, 
    we observe that the lemma follows once this inequality is 
    established. Then, writing
    $$
        \mathcal{G} := \goodnodes{V}(G,k) \cap V(G'), \quad
        \mathcal{B} := \badnodes{V}(G,k) \cap V(G'), \quad
        \mathcal{N} := \neutralnodes{V}(G,k) \cap V(G'),
    $$
    we bound $|\mathcal{G}|$ from below, and bound $|\mathcal{B}|$ from above. Finally, we double-count the edges of $G'$ in two ways, and make use of
    these bounds on $\mathcal{G}$ and $\mathcal{B}$ to prove that $|Y^*|-|X^*|$ is non-negative which gives the desired 
    inequality $|X^*| \leq |Y^*|$.

    \medskip
    We start with the reduction. 
    \begin{claim*}
         $|X^*|\leq |Y^*|$ implies that there exists a Hall graph 
         $H$ in $G$.
    \end{claim*}
    By \cref{lem:hall_graph_exists_when_|X|_less_than_|Y|}, since 
    $G'$ has no isolated nodes, there is a non-empty Hall graph $H$ 
    in $G'$. Moreover, since $N_G(Y^*) \subseteq X^*$, every 
    $y \in Y(H) \subseteq Y^*$ satisfies $N_G(y) = N_{G'}(y)$, so 
    $N_G(Y(H)) = N_{G'}(Y(H)) \subseteq X(H)$ and $H$ is a Hall 
    graph in $G$ as well.  We now move on to bound $\mathcal{G}$ and $\mathcal{B}$.
    \begin{claim*}
        $\displaystyle |\mathcal{G}| \;\geq\; \frac{|B''|}{\Delta^{\alpha+2}}$.
    \end{claim*}
    We first observe that every vertex of $C$ lies within distance 
    $\alpha+1$ of a vertex in $\goodnodes{V}(G,k)$. Indeed, let 
    $w \in V(C)$. If $w \in Y$, then \Cref{property:net_reach} for $G$ ensures 
    a vertex $u \in \goodnodes{V}(G,k)$ with $d_G(w, u) \leq \alpha$, 
    and clearly $u \in V(C)$. If $w \in X$ and 
    $w \in \goodnodes{V}(G,k)$, take $u = w$. Otherwise $w \in X$ is 
    not $k$-good in $G$, so $\degree{G}{w} = k-1$ by 
    \Cref{def:good_bad_neutral}, which is at least $1$ since $k \geq 2$. 
    Hence $w$ has a neighbor $y' \in Y \cap V(C)$, and applying 
    \Cref{property:net_reach} for $G$ to $y'$ produces 
    $u \in \goodnodes{V}(G,k) \cap V(C)$ with 
    $$
        d_G(w, u) \;\leq\; d_G(w, y') + d_G(y', u) 
                 \;\leq\; 1 + \alpha.
    $$

    Now, for each $w \in B''$ choose a vertex 
    $v(w) \in \goodnodes{V}(G,k)$ with $d_G(w, v(w)) \leq \alpha+1$, 
    which exists just as we saw now. Since $s = 4\alpha+6$, by triangle inequality
    $$
        d_G(x, v(w)) \ \leq\ d_G(x,w) + (\alpha+1) 
                     \ \leq\ (t-s) + (\alpha+1) 
                     \ =\    t - 3\alpha - 5 
                     \ <\   t,
    $$
    so $v(w) \in V(G')$ and hence $v(w) \in \mathcal{G}$. We now count 
    how many $w \in B''$ can share the same image $v(w) = v$. Any such 
    $w$ satisfies $d_G(w, v) \leq \alpha+1$, so it lies in 
    $B_G(v, \alpha+1)$; and a vertex with maximum degree $\Delta$ has 
    $$
        |B_G(v, \alpha+1)| \;\leq\; \sum_{i=0}^{\alpha+1} \Delta^i 
                           \;\leq\; \Delta^{\alpha+2}
    $$
    since $\Delta \geq 2$. Therefore at most $\Delta^{\alpha+2}$ 
    vertices of $B''$ are mapped to any single $v \in \mathcal{G}$, and
    $$
        |\mathcal{G}| \;\geq\; \frac{|B''|}{\Delta^{\alpha+2}}.
    $$
    The claim is hence true.
    
    \medskip
    
    We now bound $|\mathcal{B}|$ from above. 
    \begin{claim*}
        $\displaystyle |\mathcal{B}| \;\leq\; 1 
        + \tfrac{20}{\beta}\,|\mathcal{G}| 
        + \epsilon\,|B''|$.
    \end{claim*}
    Set $\ell := \lceil \beta\alpha/4 \rceil$; since $\beta \geq 60$ 
    and $\alpha \geq 3$, we get $\ell \geq 15\alpha$. We partition $\mathcal{B}$ 
    by distance from $x$ into three regions:
    \begin{align*}
        \mathcal{B}_{\operatorname{near}} 
            &:= \{b \in \mathcal{B} : d_G(x,b) < \ell\}, \\
        \mathcal{B}_{\operatorname{core}} 
            &:= \{b \in \mathcal{B} : \ell \leq d_G(x,b) \leq t-(\alpha+1)\}, \\
        \mathcal{B}_{\operatorname{shell}} 
            &:= \{b \in \mathcal{B} : d_G(x,b) > t-(\alpha+1)\}.
    \end{align*}
    Any two distinct $b_1, b_2 \in \mathcal{B}_{\operatorname{near}}$ would 
    satisfy
    $$
        d_G(b_1, b_2) \;\leq\; d_G(b_1, x) + d_G(x, b_2) 
                     \;<\;    2\ell 
                     \;<\;    \beta\alpha,
    $$
    contradicting \Cref{property:net_distance} of $G$ since both lie in 
    $\badnodes{V}(G,k)$. Hence $|\mathcal{B}_{\operatorname{near}}| \leq 1$. 
    Moreover, since $\alpha+1 < s$, the definition of 
    $\mathcal{B}_{\operatorname{shell}}$ implies $d_G(x,b) > t-s$ for any $b\in \mathcal{B}_{\operatorname{shell}}$, so 
    $\mathcal{B}_{\operatorname{shell}} \subseteq V(G') \setminus B''$ and 
    therefore
    $$
        |\mathcal{B}_{\operatorname{shell}}| \;\leq\; |G'| - |B''| 
                                       \;\leq\; \epsilon\,|B''|
    $$
    by the choice of $t$.

    It remains to bound $|\mathcal{B}_{\operatorname{core}}|$. Fix $b \in \mathcal{B}_{\operatorname{core}}$ 
    and a shortest path $P_b$ from $x$ to $b$ in $G$. Since 
    $d_G(x,b) \geq \ell$, the last $\ell+1$ vertices of $P_b$ form a 
    path $b = u_0, u_1, \dots, u_\ell$ with $d_G(x, u_i) = d_G(x,b) - i$, 
    so every $u_i \in V(G')$. Set 
    $M := \lfloor \ell/(2\alpha+3) \rfloor$ and $i_j := j(2\alpha+3)$ 
    for $0 \leq j \leq M$, so each $i_j \in \{0, \dots, \ell\}$. By the 
    observation proved at the start of the previous claim, for each 
    $j$ there exists $g_j(b) \in \goodnodes{V}(G,k)$ with 
    $d_G(u_{i_j}, g_j(b)) \leq \alpha+1$; then
    $$
        d_G(x, g_j(b)) \;\leq\; d_G(x, u_{i_j}) + (\alpha+1) 
                       \;\leq\; d_G(x, b) + (\alpha+1) 
                       \;\leq\; t,
    $$
    so $g_j(b) \in \mathcal{G}$. Set 
    $\mathcal{G}(b) := \{g_0(b), \dots, g_M(b)\}$. We claim that 
    $|\mathcal{G}(b)| = M+1$ and that the sets $\mathcal{G}(b)$ are 
    pairwise disjoint as $b$ ranges over $\mathcal{B}_{\operatorname{core}}$.

    For distinctness within $\mathcal{G}(b)$: if $g_j(b) = g_{j'}(b)$ 
    with $j \neq j'$, then $d_G(u_{i_j}, u_{i_{j'}}) \leq 2(\alpha+1)$, 
    while $u_0, \dots, u_\ell$ is a subpath of a shortest path and so 
    $d_G(u_{i_j}, u_{i_{j'}}) = |i_j - i_{j'}| \geq 2\alpha+3$, a 
    contradiction. For disjointness across 
    $\mathcal{B}_{\operatorname{core}}$: if 
    $g \in \mathcal{G}(b) \cap \mathcal{G}(b')$ with $b \neq b'$, then, 
    using $d_G(b, u_{i_j}) = i_j \leq \ell$,
    $$
        d_G(b, b') \;\leq\; d_G(b, u_{i_j}) + d_G(u_{i_j}, g) 
                            + d_G(g, u'_{i_{j'}}) + d_G(u'_{i_{j'}}, b') 
                  \;\leq\; 2\ell + 2(\alpha+1) 
                  \;<\;    \beta\alpha,
    $$
    contradicting \Cref{property:net_distance} since both $b, b'$ lie 
    in $\badnodes{V}(G,k)$. Since 
    $2\alpha+3 \leq 5\alpha$ and $\ell=\lceil \beta\alpha/4 \rceil$, we have 
    $$
        M + 1 \;\geq\; \frac{\ell}{2\alpha+3} 
              \;\geq\; \frac{\ell}{5\alpha} 
              \;\geq\; \frac{\beta}{20}.
    $$
    The sets $\mathcal{G}(b)$ are therefore pairwise disjoint subsets 
    of $\mathcal{G}$, so 
    $$
        |\mathcal{G}| \;\geq\; (M+1)\,|\mathcal{B}_{\operatorname{core}}| 
                      \;\geq\; \tfrac{\beta}{20}\,|\mathcal{B}_{\operatorname{core}}|,
    $$
    which gives the inequality
    $$
    |\mathcal{B}_{\operatorname{core}}| \leq \frac{20}{\beta} |\mathcal{G}|.
    $$
    Summing over the upper bounds on $|\mathcal{B}_{\operatorname{core}}|$, $|\mathcal{B}_{\operatorname{shell}}|$ and $|\mathcal{B}_{\operatorname{near}}|$ yields the claim.

    \medskip
    Having bounded $|\mathcal{G}|$ from below and $|\mathcal{B}|$ from 
    above, we proceed to the edge count that links these bounds to the  
    term $|Y^*| - |X^*|$. Set 
    $$
        X^*_{\operatorname{int}} := X^* \cap B_G(x, t-2), \qquad
        X^*_{\operatorname{bd}}  := X^* \setminus X^*_{\operatorname{int}}.
    $$
    For $x' \in X^*_{\operatorname{int}}$, every $G$-neighbor of $x'$ lies at 
    distance at most $t-1$ from $x$ and hence in $Y^*$, so 
    $\degree{G'}{x'} = \degree{G}{x'}$. Since $G$ is a \net{k}{\alpha}{\beta}, $\degree{G}{x'} \geq k-1$ for any vertex of $X$ , with $\degree{G}{x'} \geq k$ exactly 
    when $x' \in \goodnodes{V}(G,k)$. Counting edges from the 
    $X$-side, 
    $$
        |E(G')| \;=\; \sum_{x' \in X^*} \degree{G'}{x'} 
                \;\geq\; \sum_{x' \in X^*_{\operatorname{int}}} \degree{G}{x'} 
                \;\geq\; (k-1)|X^*_{\operatorname{int}}| 
                         + |X^*_{\operatorname{int}} \cap \goodnodes{V}(G,k)|.
    $$
    From the $Y$-side, $N_G(Y^*) \subseteq X^*$ gives 
    $\degree{G'}{y} = \degree{G}{y}$ for every $y \in Y^*$, where 
    $\degree{G}{y} \leq k-2$ if $y \in \goodnodes{V}(G,k)$, 
    $\degree{G}{y} = k$ if $y \in \badnodes{V}(G,k)$, and 
    $\degree{G}{y} = k-1$ if $y \in \neutralnodes{V}(G,k)$. Therefore
    $$
        |E(G')| \;=\; \sum_{y \in Y^*} \degree{G}{y} 
                \;\leq\; (k-1)|Y^*| - |Y^* \cap \goodnodes{V}(G,k)| 
                         + |Y^* \cap \badnodes{V}(G,k)|.
    $$
    Comparing the two inequalities and rearranging,
    \begin{equation} \label{eqn:eq_1}
         (k-1)\bigl(|Y^*| - |X^*_{\operatorname{int}}|\bigr) 
        \;\geq\; |X^*_{\operatorname{int}} \cap \goodnodes{V}(G,k)| 
                 + |Y^* \cap \goodnodes{V}(G,k)| 
                 - |Y^* \cap \badnodes{V}(G,k)|.
    \end{equation}
    Since $\mathcal{G}\subseteq X^*\sqcup Y^*$, we have that 
    \begin{equation} \label{eqn:eq_2}
         |X^*_{\operatorname{int}} \cap \goodnodes{V}(G,k)| 
        + |Y^* \cap \goodnodes{V}(G,k)| 
        \;\geq\; |\mathcal{G}| - |X^*_{\operatorname{bd}}|;
    \end{equation}
    and, since every $k$-bad vertex in $G$ lies in $Y$ by \Cref{def:net},
    \begin{equation} \label{eqn:eq_3}
         |Y^* \cap \badnodes{V}(G,k)| = |\mathcal{B}|.
    \end{equation}
    Adding 
    $-(k-1)|X^*_{\operatorname{bd}}|$ to both sides in \Cref{eqn:eq_1} and substituting the bounds from \Cref{eqn:eq_2,eqn:eq_3} gives
    $$
        (k-1)\bigl(|Y^*| - |X^*|\bigr) 
        \ \geq\ |\mathcal{G}| - |\mathcal{B}| - k\,|X^*_{\operatorname{bd}}|.
    $$

    \medskip
    It now remains to substitute the bounds on $|\mathcal{G}|$, 
    $|\mathcal{B}|$, and $|X^*_{\operatorname{bd}}|$ and verify that the 
    right-hand side is non-negative. By the bound on $|\mathcal{B}|$,
    $|\mathcal{B}| \leq 1 + \tfrac{20}{\beta}|\mathcal{G}| 
    + \epsilon|B''|$. Also $X^*_{\operatorname{bd}} \subseteq V(G') \setminus B''$ 
    since $t-2 \geq t-s$, so $|X^*_{\operatorname{bd}}| \leq \epsilon|B''|$ 
    by the choice of $t$. Using $\beta \geq 60$,
    $$
        (k-1)\bigl(|Y^*| - |X^*|\bigr) 
        \ \geq\ |\mathcal{G}| - |\mathcal{B}| - k\,|X^*_{\operatorname{bd}}| \ \geq\ \Bigl(1 - \tfrac{20}{\beta}\Bigr)|\mathcal{G}| 
                 - 1 - (k+1)\epsilon|B''| 
         \ \geq\ \tfrac{2}{3}|\mathcal{G}| - 1 - (k+1)\epsilon|B''|.
    $$
    The lower bound on $|\mathcal{G}|$ ( that 
    $|\mathcal{G}| \geq |B''|/\Delta^{\alpha+2}$) and the assumption 
    $(k+1)\epsilon \leq 1/(3\Delta^{\alpha+2})$ from the lemma statement yields
    $$
        (k-1)\bigl(|Y^*| - |X^*|\bigr) 
        \;\geq\; \frac{|B''|}{3\Delta^{\alpha+2}} - 1.
    $$
    Finally, \rad{C} is at least $t \geq t_0$, so \ecc{C}{x} is at least $t_0$ and $C$ contains vertices at 
    every distance from $x$ between $0$ and $t_0$, giving 
    $|B''| \geq (t_0 - s) + 1$. Since $t_0 > 10\Delta^{\alpha+2}$ and 
    $s = 4\alpha+6$, $|B''| \geq 3\Delta^{\alpha+2}$, and the 
    right-hand side is non-negative. As $k \geq 2$, this gives 
    $|X^*| \leq |Y^*|$ and we are done.
\end{proof}
the above inequality with the fact that $k\geq 1$ and $\alpha\geq 3$, we get that

We now combine 
\Cref{lem:weaknet_hall,lem:residual-is-weaknet-or-net,lem:slow-growing-ball} 
to prove the central lemma of this section: every $x \in X$ is contained in the $X$-side of a Hall graph lying in a small ball around $x$. The construction iterates the following: at each step, either the 
component of the current residual graph containing $x$ has small 
radius and serves as the next Hall graph (via 
\Cref{lem:weaknet_hall,lem:residual-is-weaknet-or-net}), or 
\Cref{lem:slow-growing-ball} produces a Hall graph in a small ball 
around $x$. The iteration halts as soon as $x$ lies in the $X$-side 
of the Hall graph just produced, and terminates in finitely many 
steps because each removal strictly shrinks $|X|$. 
\Cref{lem:union_of_hall_graphs_is_hall} then glues the resulting 
Hall graphs into a single one.
\Cref{lem:union_of_hall_graphs_is_hall} then glues the
resulting Hall graphs into a single one.

The overall shape of the above construction--iteratively removing 
Hall graphs from the current residual graph until one containing $x$ 
on its $X$-side is produced--follows \cite{brandt2025locality}. 
In their setting, the residual graphs retain the same strong 
structural properties as the original graph, so the same lemma 
applies at every step with no additional work. In our weaker setting, 
removing a Hall graph from a net degrades its structure: the 
residual graph is no longer a net, only a weak-net, and its 
components are nets with weaker parameters. Controlling this 
degradation is exactly the role of the structural lemmas 
\Cref{lem:residual-reach} and 
\Cref{lem:residual-is-weaknet-or-net}, which ensure that the 
argument of \Cref{lem:slow-growing-ball} can be reapplied 
inside the iteration.

\begin{lemma}\label{lem:hall_graph_containing_node_in_net}
    Let $G=(X,Y,E)$ be a \net{k}{\alpha}{\beta} on $n$ vertices with $k\geq 2$, 
    $\alpha\geq 3$, $\beta \geq 180$, and maximum degree $\Delta\geq 2$. 
    Then, for $x\in X$ and 
    $\kappa(n) := 600 k\alpha \Delta^{3\alpha+2}\log n + 4\beta\alpha$, 
    there exists a Hall graph $H$ in $G$ such that $x \in X(H)$ and 
    $V(H) \subseteq B_G(x, \kappa(n))$.
\end{lemma}

\begin{proof}
      We build $H$ as a finite union of Hall graphs $H_1, \ldots, H_m$ in 
    a sequence of bipartite graphs $G_0 := G$ and 
    $G_{i+1} := G_i \setminus H_{i+1}$, halting once $x \in X(H_m)$. At 
    each step the connected component of $G_i$ containing $x$ either has 
    radius at most $\frac{\kappa(n)}{2}$, in which case it serves as $H_{i+1}$ directly, 
    or has radius at least $\frac{\kappa(n)}{2}$ in which case \Cref{lem:slow-growing-ball} extracts a Hall 
    graph from $B_{G_i}(x,\kappa(n))$. Each $H_j$ 
    is non-empty, so $|X(G_i)|$ strictly decreases and the process 
    terminates; and each $H_j$ satisfies $V(H_j) \subseteq B_{G_{j-1}}(x, \kappa(n))$, 
    hence $V(H_j) \subseteq B_G(x, \kappa(n))$.

   Set $\alpha' := 3\alpha$ and $\beta' := \lfloor\beta/3\rfloor$, and set
    $t_0 := 2(\beta' \alpha' + 10 \Delta^{\alpha' + 2})
         = 2(\lfloor\beta/3\rfloor \cdot 3\alpha + 10 \Delta^{3\alpha + 2})$,
    $s := 4\alpha' + 6 = 12\alpha + 6$,
    $\epsilon := 1/(3(k+1)\Delta^{\alpha'+2}) 
              = 1/(3(k+1)\Delta^{3\alpha+2})$, and
    $L := \lceil \log_{1+\epsilon}(n) \rceil + 1$. Observe that
    $$
        L-1=\lceil \log_{\bigl(1 + \frac{1}{3(k+1)\Delta^{3\alpha+2}}\bigr)} n \rceil 
        \ \leq \ \lceil 9(k+1)\Delta^{3\alpha+2} \log n \rceil 
        \ \leq \ 20 k \Delta^{3\alpha+2} \log n.
    $$
    Combining the above inequality with the facts that $k \geq 2$ and 
    $\alpha \geq 3$, and using $\lfloor\beta/3\rfloor \cdot 3 \leq \beta$,
    \begin{align*}
        t_0 + s(L-1)\ 
        &\leq\ 2\bigl(\lfloor\beta/3\rfloor \cdot 3\alpha 
               + 10\Delta^{3\alpha+2}\bigr) 
               + (12\alpha + 6)\bigl(20 k \Delta^{3\alpha+2} \log n\bigr) \\
        &\leq\ 2\beta\alpha + 20\Delta^{3\alpha+2} 
               + 280 k\alpha \Delta^{3\alpha+2} \log n  \\
        &\leq\ 2\beta\alpha + 300 k\alpha \Delta^{3\alpha+2} \log n \\
        &=\ \frac{\kappa(n)}{2}.
    \end{align*}

    Given $G_i$ with $x \in X(G_i)$, let $C_i$ denote the connected 
    component of $G_i$ containing $x$. By 
    \Cref{lem:residual-is-weaknet-or-net} applied with $j = i$, $C_i$ 
    is a \weaknet{k}{\alpha}{\beta}, and either has radius at most 
    $\beta\alpha/2 \leq \kappa(n)$ or is a 
    \net{k}{3\alpha}{\lfloor \beta/3 \rfloor}.

    Suppose first that \rad{C_i} is at most $t_0+s(L-1)$. Then 
    $V(C_i) \subseteq B_{G_i}(x, \kappa(n))$, and since $C_i$ is a 
    \weaknet{k}{\alpha}{\beta}, \Cref{lem:weaknet_hall} gives an 
    $X(C_i)$-saturating matching in $C_i$. Combined with 
    $N_{G_i}(Y(C_i)) \subseteq X(C_i)$, which follows from $C_i$ being 
    a connected component of $G_i$, this makes $H_{i+1} := C_i$ a Hall 
    graph in $G_i$ with $x \in X(H_{i+1})$ and 
    $V(H_{i+1}) \subseteq B_{G_i}(x, \kappa(n))$. Halt with $m := i+1$.

    Suppose instead that \rad{C_i} is greater than $t_0+s(L-1)$, so by 
    \Cref{lem:residual-is-weaknet-or-net} it is a 
    \net{k}{3\alpha}{\lfloor \beta/3 \rfloor}. Since \rad{C_i} exceeds $t_0 + s(L-1)\geq \frac{\beta\alpha}{2}$, applying 
    \Cref{lem:slow-growing-ball} to $C_i$ at $x$ produces a non-empty 
    Hall graph $H_{i+1}$ in $G_i$ with 
    $V(H_{i+1}) \subseteq B_{G_i}(x, t_0 + s(L-1)) \subseteq B_{G_i}(x, \kappa(n))$. 
    If $x \in X(H_{i+1})$, halt with $m := i+1$; otherwise pass to 
    $G_{i+1}$ and continue.

    Each $H_{i+1}$ is non-empty, so $|X(G_{i+1})| < |X(G_i)|$; this 
    strictly decreasing sequence of non-negative integers forces the 
    process to halt in at most $|X|$ iterations, and since 
    $x \in X(G_i)$ until the halting step, halting is reached with 
    $x \in X(G_{m-1})$.

    It remains to verify that $H := \bigcup_{j=1}^{m} H_j$ satisfies 
    the three required conditions. By construction $H_1$ is a Hall 
    graph in $G_0 = G$ and, for $2 \leq j \leq m$, $H_j$ is a Hall 
    graph in $G \setminus \bigcup_{i<j} H_i$, so $H$ is a Hall graph 
    in $G$ by \Cref{lem:union_of_hall_graphs_is_hall}. The halting condition gives 
    $x \in X(H_m) \subseteq X(H)$. Finally, for every $v \in V(H_j)$ we have $v \in B_{G_{j-1}}(x, \kappa(n))$, i.e.\ $d_{G_{j-1}}(x, v) \leq \kappa(n)$. By \Cref{obsvn:residual_graph_properties}, $G_{j-1}$ is an induced subgraph of $G$, so $d_G(x, v) \leq d_{G_{j-1}}(x, v) \leq \kappa(n)$. Hence $V(H) \subseteq B_G(x, \kappa(n))$.
\end{proof}

\subsection{Algorithm for Saturating Matching in Nets}
With \Cref{lem:hall_graph_containing_node_in_net} in hand, we are ready to prove the main theorem of the section. The algorithm has each $x \in X$ construct a local Hall graph $H_x$ containing it, together with a $X(H_x)$-saturating matching. Using a consistent tie-breaking rule, each $x\in X$ selects a unique Hall graph in which it is contained and consequently a unique ``local matching'' edge in which it is contained. Showing that these local matchings together form a global matching gives an $X$-saturating matching. This is virtually the same argument from \cite{brandt2025locality}.
\begin{theorem}\label{thm:Density-HSO}
Let $k \geq 2$ be an integer and let $\alpha, \beta$ be positive integers with $\alpha \geq 3$ and $\beta \geq 180$. 
There is a deterministic LOCAL algorithm that, given a \net{k}{\alpha}{\beta} $G = (X, Y, E)$ on $n$ vertices with maximum degree $\Delta\geq 2$, computes an $X$-saturating matching in $O(k\alpha \Delta^{3\alpha+2}\log n +\beta\alpha)$ rounds.
\end{theorem}
\begin{proof}
    Set $r := 600 k\alpha \Delta^{3\alpha+2} \log n + 4\beta\alpha$. 
    We now describe the algorithm.

    In the first $r+1$ rounds, each $x \in X$ collects the ball $B_G(x, r+1)$. 
    By \Cref{lem:hall_graph_containing_node_in_net}, there exists a 
    Hall graph $H$ in $G$ with $x \in X(H)$ and 
    $V(H) \subseteq B_G(x, r)$. For any candidate subgraph $H$ with 
    $V(H) \subseteq B_G(x, r)$, both defining conditions of 
    \Cref{def:Hall_graph} --- the existence of an $X(H)$-saturating 
    matching, and $N_G(Y(H)) \subseteq X(H)$ --- depend only on the 
    subgraph of $G$ induced on $B_G(x, r+1)$, since every 
    $G$-neighbor of a vertex of $V(H)$ lies in $B_G(x, r+1)$. Hence 
    $x$ can enumerate all candidate subgraphs and test both 
    conditions by local computation. Each vertex $x$ sets $(H_x, M_x)$ to 
    be the first pair, in some fixed canonical order (say, 
    lexicographic with respect to the identifier sequences of the 
    candidates), consisting of a Hall graph $H_x$ in $G$ with 
    $x \in X(H_x)$ and $V(H_x) \subseteq B_G(x, r)$ and an 
    $X(H_x)$-saturating matching $M_x$ of $H_x$. In the next $r+1$ 
    rounds, each $x \in X$ collects $(H_{x'}, M_{x'})$ from every 
    $x' \in B_G(x, r+1) \cap X$.

    For each $x \in X$, set
    $$
        \mathcal{H}(x) := \bigl\{ x' \in X : x \in X(H_{x'}) \bigr\}.
    $$
    We record two properties of $\mathcal{H}(x)$. First, 
    $\mathcal{H}(x)$ is non-empty, since $x \in \mathcal{H}(x)$. 
    Second, every $x' \in \mathcal{H}(x)$ satisfies 
    $d_G(x, x') \leq r$, because 
    $x \in V(H_{x'}) \subseteq B_G(x', r)$. In particular, $x$ has 
    collected $(H_{x'}, M_{x'})$ for every $x' \in \mathcal{H}(x)$.

    Each $x \in X$ now sets $x^\star$ to be the vertex of 
    $\mathcal{H}(x)$ with the largest identifier, and lets $e_x$ be 
    the edge of $M_{x^\star}$ incident on $x$. This edge exists and 
    is unique since $x \in X(H_{x^\star})$ and $M_{x^\star}$ is an 
    $X(H_{x^\star})$-saturating matching. Let 
    $M := \{e_x : x \in X\}$ denote the resulting set of edges.

    We claim that $M$ is an $X$-saturating matching. Suppose for the 
    sake of contradiction that $M$ is not a matching. Then there 
    exist distinct $u, v \in X$ and a vertex $z \in Y$ with 
    $e_u \cap e_v = \{z\}$.

    The vertex $z$ is an endpoint of $e_u \in M_{u^\star}$, so 
    $z \in Y(H_{u^\star})$. Similarly, $z \in Y(H_{v^\star})$. Since 
    $H_{u^\star}$ and $H_{v^\star}$ are Hall graphs in $G$, 
    \Cref{def:Hall_graph} gives $N_G(z) \subseteq X(H_{u^\star})$ 
    and $N_G(z) \subseteq X(H_{v^\star})$. Both $u$ and $v$ are 
    $G$-neighbors of $z$, hence
    $$
        \{u, v\} \subseteq X(H_{u^\star}) \cap X(H_{v^\star}).
    $$

    We distinguish two cases. Suppose first that $u^\star = v^\star$. 
    Then $e_u$ and $e_v$ are two edges of $M_{u^\star}$ sharing the 
    vertex $z$, forcing $e_u = e_v$. An edge has exactly two 
    endpoints, one of which is $z$; since both $u$ and $v$ are 
    endpoints distinct from $z$, we obtain $u = v$, contradicting 
    $u \neq v$.

    Suppose next that $u^\star \neq v^\star$. Without loss of 
    generality, $\mathrm{id}(v^\star) < \mathrm{id}(u^\star)$. From 
    $v \in X(H_{u^\star})$ we have $u^\star \in \mathcal{H}(v)$. 
    Then $\mathrm{id}(v^\star) < \mathrm{id}(u^\star)$ contradicts 
    the choice of $v^\star$ as the vertex of $\mathcal{H}(v)$ with 
    the largest identifier.

    Both cases yield a contradiction, so $M$ is a matching. 
    Moreover, every $x \in X$ is incident to $e_x \in M$, so $M$ is 
    $X$-saturating. The algorithm runs in $2(r+1) = 
    O(k\alpha \Delta^{3\alpha+2} \log n + \beta\alpha)$ rounds.
\end{proof}

\section{A $(\Delta+1)$-Edge Coloring Algorithm}\label{sec:edge-col}

Let $G=(V, E)$ be a bipartite graph whose bipartition is unknown to the nodes. In this section, we present an algorithm that computes an $(\Delta + 1)$-edge coloring of $G$ in $O(f(\Delta) \log n)$ rounds. We begin by introducing some definitions and algorithmic subroutines that are used during the execution of the algorithm. 

The first definition we introduce is the bipartite double cover. Given any graph $G$, its bipartite double cover is a a virtual bipartite graph obtained by creating two copies of every node in $G$, one copy on each side of the bipartition. Each copy is adjacent to the copies of the neighbors of its original node on the opposite side. By construction, every node of $G$ knows on which side of the bipartition each of its copies is located, hence when describing an algorithm on the bipartite double cover, we can assume that the nodes of the bipartite double cover know on which side of the bipartition they lie. 

\begin{definition}[Bipartite Double Cover]
    \label{def:bipdc}
    The bipartite double cover of the graph $G=(V,E)$ is a graph $G_b=(X,Y,E_b)$ where for every node $v\in V$ there are nodes $v_X \in X$ and $v_Y \in Y$ and for every edge $(u,v)\in E$, there is an edge $(u_X,v_Y)$ and an edge $(v_X,u_Y)$.
\end{definition}

We also recall the following folklore result for the bipartite double cover of bipartite graphs. 

\begin{lemma}\label{lem:bipdc}
    The bipartite double cover of a bipartite graph $G$ consist of two disjoint copies of the graph $G$.
\end{lemma}

We will make use of two further results during the execution of our algorithm.
The first is an existential result about edge coloring on bipartite graphs that will be needed when we take care of coloring small connected components, which we handle with a brute-force approach.

\begin{lemma}[Bipartite edge coloring \cite{Konig1916,BorodinKW97}]
\label{lem:bipedgecolor}
    Every bipartite graph of maximum degree $\Delta$ admits an edge coloring using $\Delta$ colors. Even stronger, every bipartite graph can be edge colored such that for every edge $e=(u,v)$ the color of $e$ is in the set $\{1,...,\max(\deg(v),\deg(u))\}$.
\end{lemma}

The second result is a result about the sinkless orientation problem, which consists in orienting the edges of the input graph such that every node has one outgoing edge. There is an algorithm solving this problem deterministically in $O(\log n)$ rounds for graphs of minimum degree $3$. This sinkless orientation algorithm will be used as a subroutine in the algorithm for \Cref{lem:node-picking}. The sinkless orientation algorithm also works for half-edges which are edges with only one endpoint.

\begin{lemma}[Sinkless orientation, Theorem 1.2 of \cite{GhaffariS17}]
    \label{lem:sinklesso}
    For any graph $G$ with minimum degree at least $3$, possibly with half-edges, there is a deterministic $O(\log n)$-round algorithm orienting the edges and half-edges of the graph such that every node has an outgoing edge or half-edge.
\end{lemma}

Now that we have introduced the lemmas and definition from the literature that we require in this section, we introduce a new problem that we need to solve in our $(\Delta+1)$-edge coloring algorithm.
We call this problem the \emph{Node Selection problem}.

\begin{definition}[Node Selection]
    Let $G$ be a graph and $\mathcal{C}$ a clustering associated with this graph such that every node $v$ in $G$ is located in one single cluster. In the Node Selection problem the task is to select one node from each cluster such that no two selected nodes are neighbors in $G$.
\end{definition}

Assume that the diameter of each cluster is both greater than $3^{4\Deltaball}$ and in $O(f(\Delta))$ for some function $f$. 
Then we propose an $O(f(\Delta)\log n)$-round algorithm solving the Node Selection problem.

\begin{restatable}{lemma}{NodeSelection}[Node Selection]
\label{lem:node-picking}
    Let $G$ a graph of maximum degree $\Deltaball$ and $\mathcal C$ a partition of the nodes of the graph into subset of nodes such that every $b\in \mathcal C$ has diameter at least $d_{\min} \geq 3^{4\Deltaball}$. The \NodePicking$(G, \mathcal C, d_{\max}, \Deltaball)$ algorithm returns a subset $P$ of the nodes in $G$ such that, 
    \begin{enumerate}
        \item At Most One Node Per Cluster: There is no $C\in \mathcal C$ and $v_p,u_p\in P$, such that $v_p\in C$,and $u_p\in C$
        \item At Least One Node Per Cluster: For every $C\in \mathcal C$, there exists $v_p\in P$, such that $v_p\in C$.
        \item No Neighboring Selected Nodes: For every $v_p,u_p\in P$, $d_G(v_p,u_p) \geq 2$
    \end{enumerate}
    The round complexity of \NodePicking$(G, \mathcal C, d_{\max}, \Deltaball)$ is $O(\Deltaball \cdot d_{\max} \cdot \log n)$
\end{restatable}

We postpone the proof of \Cref{lem:node-picking} to \Cref{sec:node-selection}.

\subsection{Description of the Algorithm}

\begin{algorithm}[htb!]
\caption{\EdgeColoring ($G=(V,E)$)}
\label{alg:main-alg}
\begin{algorithmic}[1]
\State Let $G_1$ be a copy of $G$ and let $d_{\Delta-2} = 100$
\For{$i=\Delta-2$ \textbf{down to} $1$}
    \State Let $d_{i-1} = 16\cdot 11520\cdot 3^{4\Delta ^{d_{i}}}$ 
\EndFor
\For{$i=1$ to $\Delta-2$}
    \State Let $d_d = d_{i-1}/11520$ and let $M_i = \emptyset$
    \State Remove from $G_i$ every connected component of diameter less than $16d_{i-1}$, and color these components (\Cref{sec:smalldiam})
    \State Let $G_b=(X,Y,E_b)$ be the bipartite double cover of $G_i$, from which we remove from $X$ every node of degree strictly less than $\Delta - i +1$
    \State Let $G_U$ be the subgraph of $G_i$ induced by all nodes at distance less than $16d_d$ from a node of degree $\Delta - i + 2$
    \State Compute an MIS $S$ on $(G_i\setminus G_U)^{8d_d}$ 
    \State For every node $u\in S$, add to $S$ a node at distance $2d_d + 1$ from $u$ in $G_i \setminus G_U$
    \State For every node $v$ of degree $\Delta-i+2$, add a node at distance $2d_d$ from $v_Y$; if no such node exists, match the connected component containing $v_Y$ and add this matching to $M_i$
    \For{\textbf{each} $u \in S$}
        \State Add a copy $u'_Y$ of $u_Y$ to $Y$
    \EndFor
    \State Compute an $X$-saturated matching of $G_b$ (by \Cref{thm:Density-HSO}) and add it to $M_i$
    \State Let $G_{b,hso}$ be the graph on $V$ with an edge between $u$ and $v$ if $(u_X,v_Y)\in M_i$ or $(v_X,u_Y)\in M_i$
    \State Let $H = (V,\, E(G_{b,hso}))$ 
    \State Let $S_i = {}$\SparseNodePicking$(G_i,H,d_i,\Delta)$
    \State Let $M''_i$ be the matching obtained by matching every node except the duplicated nodes and the nodes in $S_i$ \Comment{\Cref{sec:almostperfectmatchingNS}}
    \State Color every edge of $M''_i$ with color $\Delta - i +2$ \Comment{\Cref{sec:mainlemma}}
    \State Remove $M_i$ to obtain $G_{i+1}$
\EndFor
\State Edge-color the graph $G_{\Delta-1}$ using the colors $\{1,2,3\}$ (\Cref{lem:colordegree3graph})
\end{algorithmic}
\end{algorithm}

\begin{algorithm}[htb!]
\caption{\SparseNodePicking$(G_i,H,d_i,\Delta)$}
\label{alg:sparse-node-picking}
\begin{algorithmic}[1]
\State Let $\Delta' = \Delta^{2d_i}$ and let $d_{c}= 4 \cdot 3^{4\Delta'}$
\State Let $H'$ be a copy of $H$
\State Remove from $H'$ every node contained in a connected component of weak diameter less than $2d_{c}$ in $G_i$
\State Remove from $H'$ every node at distance less than $2d_i$ from a node of degree $\Delta-i+2$ in $G_i$, and every duplicated node
\State Let $\mathcal{C} = {}$\BallCreation$(H',d_c)$
\State Let $H''$ be a copy of $H$ where we add an edge between two nodes if they are in different clusters and at distance less than $2d_i$ in $G_i$
\State \Return \NodePicking$(H'',\mathcal{C},d_{c}/4,d_{c}/2,\Delta')$
\end{algorithmic}
\end{algorithm}

\begin{algorithm}[htb!]
\caption{\BallCreation$(H',d_c)$}
\begin{algorithmic}[1]
\State Compute an MIS $S$ on $(H')^{d_c}$ 
\State For MIS node $s\in S$, compute $B_s$ the cluster containing all of the nodes at distance at most $d_c/4$ of $s$ in $H'$.
\State \Return $\{B_s, \text{ for each $s$ in } S \}$
\end{algorithmic}
\end{algorithm}

Let $G=(V, E)$ be a bipartite graph whose bipartition is unknown to the nodes. The algorithm first proceeds in $\Delta-2$ consecutive \emph{phases}. During phase $i$, it computes a subset of the uncolored edges and assigns them the color $\Delta-i+2$. It also fully colors the edges of the connected components of small diameter. Every phase is executed on the subgraph induced by the remaining uncolored edges. We define $G_i$ as the graph we consider in phase $i$. At the end of the $\Delta-2$ phases, the algorithm colors the subgraph induced by the remaining uncolored edges using the remaining $3$ unused colors (where colors used only for coloring isolated connected components are considered as unused colors).

At any phase $i$ of the algorithm, we maintain two invariants. The first invariant is that in the graph induced by uncolored edges, $G_i$, the degree of any node in the graph is at most $\Delta-i+2$. Maintaining this invariant ensures that after $\Delta-2$ phases of the algorithm, the graph $G_{\Delta-1}$ induced by uncolored edges has degree at most $3$, which we can then color with the three remaining colors.
The second invariant bounds the minimum distance between any two nodes of degree $\Delta-i+2$.
We maintain this distance to be at least $d_i$, which is defined via

\begin{align*}
    d_{\Delta-2} &:= 100\\
    d_{i-1} &:= 16\cdot 11520\cdot 3^{4\Delta ^{d_{i}}} \text{ for each } 1 \leq i \leq \Delta - 2.
\end{align*}
We now describe phase $i$ of the algorithm. Consider $G_i=(V,E_i)$ as the induced subgraph of $G$ composed of the edges which where not colored in the previous phases. We let $d_d = d_{i-1}/11520$ be a parameter which we will use in \Cref{sec:DoubleCover}. Let $M_i=\emptyset$, a set which will contain an intermediary matching of the nodes in $G_i$, which we use to build a virtual graph in \Cref{sec:virtualgraph}.

\begin{enumerate}
    \item We first handle the connected components of $G_i$ 
    of diameter less than $16 d_{i-1}$ in \Cref{sec:smalldiam}. Each such component is 
    bipartite with every node being of degree at most $\Delta - i + 2$, 
    and hence admits an edge coloring using colors from the palette $\{1, 
    \dots, \Delta - i + 2\}$ by \Cref{lem:bipedgecolor}. An edge colored in phase $j<i$ received either color $\Delta-j+2$ which lies outside the palette $\{1, 
    \dots, \Delta - i + 2\}$ since $j<i$, or  a color in $\{1, \dots, \Delta - j + 1\}$ if it is inside a small 
    component of $G_j$; in the latter case, all vertices of that small component were removed from the graph before phase $i$, so no such edge is adjacent to $C$. Hence the new coloring is consistent with the partial edge coloring produced so far. 
    \item We now consider $G_b$, the bipartite double cover of $G_i$ from which we remove the nodes of degree strictly less than $\Delta-i+1$ in $G_i$. The next step consists in modifying $G_b$ in order to apply \Cref{thm:Density-HSO} which we do in \Cref{sec:DoubleCover}. To apply this theorem using $G_b$, we need $G_b$ to be a $(k,\alpha,\beta)$-net as described in \Cref{def:net}. Property $1$ and $3$ of the definition are already satisfied for the parameters $(\Delta-i+2,64d_d,d_{i-1}/64d_d)$ as follows. Property $1$ is satisfied as every node in $X$ has degree at least $\Delta-i+1$ and every node in $Y$ has degree at most $\Delta-i+2$. Property $3$ is satisfied as the $(\Delta-i+2)$-bad nodes are at distance $d_{i-1}$ since the bad nodes are only the nodes of degree $\Delta-i+2$. To satisfy Property $2$, we need every node in $Y$ to be at distance $64d_d$ of a $(\Delta-i+2)$-good nodes. To achieve this, we duplicate some nodes in $Y$, so that every node in $Y$ is at distance at most $64d_d$ of a $(\Delta-i+2)$-good node as follows.
    
    We first compute $G_U$, a subgraph of $G_i$ composed of all of the nodes at distance less than $16d_d$ of a node of degree $\Delta-i+2$, because we do not want any duplicated node to be part of $G_U$. Then, to decide which nodes will be duplicated, we compute an MIS $S$ on $(G_i \setminus G_U)^{8d_d}$. Then, every node $u\in S$ picks arbitrarily one node at distance $2d_d+1$ of $u$ in $G_i$  if it exists and adds it to $S$, which is needed in order to ensure that every node in $Y$ is at distance $64d_d$ of a good node. For the same purpose, every copy $v_Y$ of a node $v$ of degree $\Delta-i+2$ adds a node at distance $2d_d$ of $v_Y$ in $G_b$ to $S$. If this is not possible and $v$ has a node of degree strictly less than $\Delta-i+1$, it implies that the connected component of $G_b$ (not $G_i$) containing $v_Y$ is of diameter at most $4d_d$. Notice that this component might not be disconnected in $G_i$ though it is disconnected in $G_b$ because the copy of $G_i$ in $G_b$ removes nodes in $X$ of degree strictly less than $\Delta-i+1$. This component does not contain any bad node as $v_Y$ degree decreased by one compared to $v$ as $v$ has a neighbor of degree stricly less than $\Delta-i+1$. In that case, we completely match every node of the component and add it to $M_i$, the set containing an intermediary matching of $G_i$. We then duplicate the $Y$ copies of the nodes in $S$, as follows. For any $v_Y$, we create $v'_Y$ such that $N(v_Y)=N(v'_Y)$. After this step $G_b$ has become a $(\Delta-i+2,64d_d,d_{i-1}/64d_d)$-net (\Cref{lem:doublenodenet}).

    \item The next step is to compute an $X$-saturated matching $M_i$ on $G_b$ and to construct a residual graph from the edge of the matching as we do in \Cref{sec:virtualgraph}. We use $M_i$ to construct a subgraph $H$ of $G_i$. The set of nodes of $H$ is the same as the set of nodes  of $G_i$. For every edge $(u_X,v_Y)\in M_i$, we add the edge $(u,v)$ with an orientation from $u$ to $v$. If both $(u_X,v_Y)$ and $(v_X,u_Y)$ are in $M_i$, we only construct one edge and do the following. If $u$ is a node of degree which was duplicated, then we orient the edge from $u$ to $v$, and similarly if $v$ is a node which was duplicated. If not, the nodes orient the edge arbitrarily.

    \item Now that $H$ has been computed, the algorithm will use this subgraph to find a matching of the nodes such that the distance between every two unmatched node is at least $d_i$. We show how to obtain this in \Cref{sec:almostperfectmatchingNS}. We show how to obtain this in $3$ steps. We let $\Delta' = \Delta^{2d_i}$ and $d_c= 4 \cdot 3^{4\Delta'}$.
    \begin{enumerate}
        \item In the first step, we remove components of weak diameter less than $2d_c$ . We can match the nodes in these component such that every node of degree at least $\Delta-i+1$ in $G_i$ is matched expect possibly nodes which were duplicated or one of their neighbors. This step takes $O(d_c)$ rounds.

        \item We now want to match the nodes left in higher diameter components. For this purpose in step $2$, we select nodes in $H$ such that the distance between two selected nodes in $H$ is at least $2d_i$ in $G_i$ and every node in $H$ is at distance at most $d_c$ of a selected node as follows. We first construct $H'$, a subgraph of $H$ from which we remove any node at distance less than $2d_i$ of a vertex of degree $\Delta-i+2$ in $G_i$ and from duplicated nodes in $G_i$. We then construct a clustering of the nodes in $H'$, which is described in \BallCreation$(H',d_c)$. We let $H''$ be a copy of $H'$ where we add an edge between two nodes if they are in different clusters and at distance less than $2d_i$ in $G_i$. The maximum degree in the graph $H''$ is at most $\Delta^{2d_i}=\Delta'$. We then apply the algorithm \NodePicking$(H'',\mathcal{C},d_{c}/4,d_{c}/2,\Delta')$ which returns one node for every cluster such that no two selected nodes are neighbors in $H''$.
        \item In step $3$, we compute a matching in $H$ using the selected nodes in the previous step. This matching is such that the only unmatched nodes are the selected nodes, the nodes of degree $3$, or one their neighbors. Also, the distance between any two unmatched nodes is at least $2d_i$. The matching is computed as follows. We first match nodes between the degree $3$ nodes and the selected nodes such that all the nodes are matched except possibly the node of degree $3$ or one of its neighbors and the selected node. Second, we match the nodes between the selected node such that only selected nodes and their neighbors can be unmatched. Lastly, we match the nodes which are just on one side of a selected node. If a selected node is left unmatched as well as one of its neighbor, we match them together. 
    \end{enumerate}
    \item We color every matched edge with color $\Delta-i+2$ in \Cref{sec:mainlemma}. We show that the distance between every two nodes of degree $\Delta-i+1$ or higher which are not matched is at least $d_i$.
\end{enumerate}

\subsection{One Phase of the Algorithm}

In this section we study the phase $i$ of the algorithm, and show the following. At the beginning of phase $i$, the graph $G{i}$ induced by uncolored edges is such that every node has a degree less or equal than $\Delta-i+2$. The colors of the adjacent edges of every node node adjacent to uncolored edges are in the set $\{\Delta-i+3,\Delta+1\}$. At the beginning of phase $i+1$, the graph $G_{i+1}$ induced by uncolored edges is such that every node has a degree less or equal than $\Delta-i+1$. The colors of the adjacent edges of every node node adjacent to uncolored edges are in the set $\{\Delta-i+2,\Delta+1\}$ .

We let $G_i$ be the graph induced by the uncolored edges in $G$. $G_i$ is a bipartite graph in which every node has a degree less or equal to $\Delta-i+2$. The distance between any two node of degree $\Delta-i+2$ in $G_i$ is at least $d_{i-1}$. Notice that in phase $1$, as there are no nodes of degree $\Delta-i+2$ in $G_1$, we can choose whichever distance $d_0$ we want as long as it is a function of $\Delta$.

\subsubsection{Removing and Coloring Connected Components of Diameter Less than $16d_{i-1}$}\label{sec:smalldiam}
Let us first look at the first step of the algorithm which consist of removing and coloring the connected component of small diameters.

\begin{lemma}
    Every connected component of $G_i$ of diameter less $16d_{i-1}$ can be removed from $G_i$ and colored in $O(d_{i-1})$ rounds.
\end{lemma}

\begin{proof}
    Every node collects its entire $16d_{i-1}$-hop neighborhood. After this, every node part of a connected component of diameter less than $16d_{i-1}$ knows its entire component. Every node in this connected component has a degree equal or less than $\Delta-i+2$ and is bipartite, so the edges in the component can be colored using colors from $1$ to $\Delta-i+2$ (\Cref{lem:bipedgecolor}). To do this, the node of maximum ID inside the component decides on the coloring of every edge and broadcast this information to every node in the component in $16d_{i-1}$ rounds.
\end{proof}

\subsubsection{Constructing A $(k,\alpha,\beta)$-net Modified Bipartite Double Cover}\label{sec:DoubleCover}
The next step is to construct a virtual graph on which we can use \Cref{thm:Density-HSO}. The theorem requires the input graph to be bipartite, with a known bipartition, and to satisfy the $(k,\alpha,\beta)$-net property, with $\alpha \geq 3$ and $\beta \geq 180$. 

We first construct the bipartite double cover $G_b=(X,Y,E_b)$ of $G_i$ and immediately remove from $X$ all the nodes of degree strictly less than $\Delta-i+1$. By construction, the bipartition of $G_b$ is already known to the nodes.
We now want to modify $G_b$ to be $(k,\alpha,\beta)$-net as described in \Cref{def:net}. The Properties $1$ and $3$ in the definition are satisfied for the parameters $(\Delta-i+2,\alpha,d_{i-1}/\alpha)$ as follows. Property $1$ is satisfied as every node in $X$ has degree at least $\Delta-i+1$ and every node in $Y$ has degree at most $\Delta-i+2$. And Property $2$ is satisfied as the $(\Delta-i+2)$-bad nodes are at distance $d_{i-1}$ since the bad nodes are only the nodes of degree $\Delta-i+2$.

In order to satisfy property $2$, we need to create some $\Delta-i+2$-good nodes. To achieve this, we first compute $G_U$ a subgraph of $G_i$ composed of all the nodes at distance less $16d_d$ of a node of degree $\Delta-i+2$, as duplicated nodes need to be at distance at least $d_i$ of nodes of degree $\Delta-i+2$. We then compute an MIS $S$ in $(G_i\setminus G_U)^{8d_d}$.  Then, every node $u\in S$ picks arbitrarily one node at distance $2d_d+1$ of $u$ in $G_i$ and add it to $S$, which is needed in order to ensure that every node in $Y$ is at distance at most $4d_d$ of a good node. Finally, we duplicate the $Y$ copies of the nodes in $S$.

We now prove some intermediate lemma which will be helpful for the next sections.

\begin{lemma}\label{lem:createdoublenodes}
let $S$ be the the set of selected nodes. the following properties are satisfied:
\begin{enumerate}
    \item The distance between every two nodes in $S$ is at least
    $2d_d$ in $G_i$.
    \item Every node in $S$ is at distance at least $2d_d$ from every node of degree $\Delta-i+2$ in $G_b$.
\end{enumerate}

This step of duplicating nodes can be done in $O(\Delta^{8d_d}\log^*n)$-rounds
\end{lemma} 

\begin{proof}
    Properties $1$ and $2$ are not useful to prove that $G_b$ is $(k,\alpha,\beta)$-net but will prove useful in the later subsections.
    Let $G_U$ be the subgraph containing all the nodes at distance less than $16d_d$ of any node of degree $\Delta - i + 2$.
    The algorithm computes an MIS $S$ on $(G_i\setminus G_U)^{8d_d}$. Then, every node in $S$ selects one node at distance $2d_d +1$ in $G_i$ and adds it to $S$. Only nodes in $S$ are possibly duplicated. 

    Property $1$ is satisfied as the distance between every node in $S$ is at least $2d_d$ in $G_i$ as follows. The distance between two MIS nodes is at least $8d_i$ and each MIS nodes select one node at distance $2d_d +1$. The distance between any two selected nodes is at least $8d_d-2(2d_d +1)> 2d_d$.

    For property $2$, we use here \Cref{lem:bipdc}, which states that the bipartite double cover of a bipartite graph consist of two disjoint copies of that graph. This implies that if two nodes in $G_i$ are at distance $k$ then their copies are at distance $k$ or are not connected. Since every node in $S$ is at distance at least $2d_d-1$ from a node of degree $\Delta-i+2$ in $G_i$, it stands in $G_b$

    The complexity of creating $S$ nodes is the complexity of computing the MIS on $(G_i\setminus G_U)^{8d_d}$ which is $O(\Delta^{8d_d}\log^* n)$-rounds.
\end{proof}

We now prove the lemma which shows that $G_b$ satisfy Property $2$ in \Cref{def:net} which we need to show that $G_b$ is $(k,\alpha,\beta)$-net. In that setting, one issue can arise when a node of degree $\Delta-i+2$ has one neighbor of degree strictly less than $\Delta-i+1$
\begin{lemma}
    \label{lem:doublecovergood}
    Every node in $Y$ is at distance at most $64d_d$ of a good node and if a node $v$ of degree $\Delta-i+2$ has a neighbor of degree strictly less than $\Delta-i+1$, then $v_Y$ adds a node to $S$ at distance $2d_d$ or it is in a connected component of diameter less or equal than $4d_d$ in $G_b$
\end{lemma}

\begin{proof}
    Let $v$ a node in $G_i$ and $v_Y\in Y$ its copy in $G_b$. $v$ is at distance at most $40d_d$ of a node $s$ in $S$. If the distance between $v_Y$ and $s_Y$ is even, then $v_Y$ and $s_Y$ are in the same component in the original double cover of $G_i$. If not, then the node selected by $s$, $s'$, at distance $2d_d+1$ of $s$, is at even distance of $v$. Then $s'$ and $v$ are in the same component. To summarize, $v_Y$ is at distance at most $64d_d$ of a node in $S$.

    Let $u$ be that node. And let $p=(v,v_2,...v_{2k},u)$ be the shortest path between $v$ and $u$ in $G_i$. If the path $p=(v_Y,v_{2,X},...v_{2k,X},u_Y)$ exists in $G_b$, it implies that $v_{2k,X}$ didn't get removed from $X$ and had degree $\Delta-i+1$. After making the copy of $u_Y$, $v_{2k,X}$ now has a degree of $\Delta-i+2$ and is a $\Delta-i+2$-good node. In that case, $v$ is at distance at most $64d_d$ of a good node.

    If this path does not exist in $G_b$, it implies that there exist $1\leq j\leq k$, such that $v_{2j,X}$ got removed. But then the degree of $v_{2j-1,Y}$ got reduced by $1$. If $v_{2j-1,Y}$ is not a node of degree $\Delta-i+2$ in $G_i$, then $v_{2j-1,Y}$ is a $\Delta-i+2$-good node in $G_b$. 
    Let's look at what happens if $v_{2j-1,Y}$ is a node of degree $\Delta-i+2$ in $G_i$. If $v_{2j-1}$ had two neighbors of degree strictly less than $\Delta-i+2$, it is now a good node. In both cases, $v_Y$ is then at distance at most $64d_d$ good node.  If $v_{2j-1,Y}$ doesn't have two neighbors with degree strictly less than $\Delta-i+2$, we do the following. $v_{2j-1,Y}$ finds one node at distance $2d_d$ in $G_b$ and adds it to $S$. This node verifies the properties of \Cref{lem:createdoublenodes} and $v$ is at distance less than $40d_d+16d_d$ of that node. 

    If $v_{2j-1,Y}$ cannot find such a node, it means that the connected component containing $v_{2j-1,Y}$ in $G_b$ is of diameter less than $4d_d$. This component doesn't contain any bad node, so a perfect matching exist in that component. Then, $v_{2j-1,Y}$ decides arbitrarily on a perfect matching for every node of the component, and we remove all of the nodes of the component from $G_b$.
    
\end{proof}

\begin{lemma}
    \label{lem:doublenodenet}
    After the duplication step, $G_b$ becomes $(\Delta-i+2,64d_d,d_{i-1}/64d_d)$-net (\Cref{def:net}). 
\end{lemma}

\begin{proof}
    The three properties that $G_b$ needs to satisfy for being $(\Delta-i+2,64d_d,d_{i-1}/64d_d)$-net are described here \Cref{def:net}.
    Property $1$ is satisfied as the minimum degree of $X$ is $\Delta-i+1$ and the degree of any node in $Y$ is at most $\Delta-i+2$. 
    Property $3$ is satisfied as the distance between every bad node is at least $d_{i-1}$ in $G_i$.
    Property $2$ is also satisfied thanks to \Cref{lem:doublecovergood}. 
\end{proof}

\subsubsection{Constructing a Virtual Graph From The X-Saturated Matching}\label{sec:virtualgraph}

The algorithm compute an $X$-saturated matching $M_i$ of $G_b$ and uses it to construct a directed subgraph $H=(V,E(H))$ of $G_i$ as following. The set of nodes of $H$ is the same as the set of nodes of $G_i$.
For every node edge $(u_X,v_Y)\in M_i$, we construct an edge $(u,v)$ in $H$. This edge is oriented from $u$ to $v$. There is one exception when there is both the edges $(u_X,v_Y)$ and $(v_X,u_Y)$ are in $M_i$. In that case, we only construct one edge. If $u$ is a node which got duplicated, the nodes oriented edge from $u$ to $v$ in $H$. Similarly form $v$ to $u$ if $v$ got duplicated. If neither $v$ or $u$ got duplicated, then the nodes orient the edge arbitrarily. 

\begin{observation}
\label{obs:graphH}
    The graph $H$ constructed from the $X$-saturated matching of $G_b$ is a directed graph with the following properties.
    \begin{enumerate}
        \item Every node of has degree at most $3$
        \item The distance between nodes of degree $3$ in $G_i$ is at least $2d_d$.
        \item Nodes with degree strictly less than $\Delta-i+1$ in $G_i$ are of degree $1$ or $0$ unless the nodes are duplicated nodes.
        \item Every node of degree $2$ which wasn't duplicated has one outgoing and one incoming edge.
        
    \end{enumerate}
    H is constructed in $O(f(\Delta)\log n)$ rounds.
\end{observation}

\begin{proof}
    Consider $G=(V,E)$ and $G_b=(X,Y,E_b)$ the bipartite double cover of $G_i$. The $X$-saturated matching is such that every node in $X$ grabs a node in $Y$, and that every node in $Y$ can be grabbed at most one, except the duplicated nodes which can be grabbed twice.
    
    We first prove property $1$. For every node $v \in V$, there is $v_X\in X$, $v_Y \in Y$ and possibly $v'_Y\in Y$ if $v$ is a duplicated node. The $X$-saturated matching ensures that $v_Y$ and $v'_Y$ are grabbed at most once. And $v_X$ grabs only one node. Thus the degree of $v$ is at most $3$.
    We now prove property $2$. To have degree $3$, a node $v$ needs to have been duplicated as non-duplicated node have degree at most $2$. Since the distance between any two duplicated nodes is at least $2d_d$ in $G_i$ by construction, the distance between degree $3$ nodes in $G_i$ is at least $2d_d$.

    Properties $3$ and $4$ come down from how the saturated matching was performed and how $H$ was constructed.
    For property $3$, nodes of degree strictly less than $\Delta-i+1$ are only contained in $Y$. Which means that their degree in $H$ is equal to $1$ or $0$ if they didn't get duplicated.
    For property $4$, we know that every node of degree $2$ which wasn't duplicated has a degree equal or higher than $\Delta-i+1$ in $G_i$. Their in $X$ got matched to a neighboring node, which gave the node an outgoing edge. And since very node can only grab once, the other edge is an incoming edge. 

    To construct $H$, the algorithm applies \Cref{thm:Density-HSO} which takes $O(k\alpha \Delta^{3\alpha+2}\log n +\beta\alpha)$ rounds for $k=\Delta-i+2$, $\alpha=16d_{d}$ and $\beta = d_{i-1}/16d_d$, which we can simplify as $O(f(\Delta)\cdot \log n)$ rounds.
    \end{proof}

\subsubsection{Almost Perfect Matching Via Node Selection}\label{sec:almostperfectmatchingNS}

The algorithm now uses the graph $H$ to find a matching of the nodes. However we cannot ensure that every node in $H$ can be matched, since we want our round complexity to be in $O(f(\Delta)\log n)$ rounds. But we want every node that which does not get matched to be at distance at least $d_i$ any other unmatched node, as this is the invariant we want to maintain.

\begin{lemma}
\label{lem:spasenodeselec}
\SparseNodePicking$(G_i,H,d_i,\Delta)$ ensures that nodes which will not be matched during phase $i$ are of degree strictly less than $\Delta-i+1$, or are of degree $\Delta-i+1$ and at distance a least $d_i$ between themselves and at distance $d_i$ of any node of degree $\Delta-i+2$ in $G_i$. This algorithm proceeds in $O(f(\Delta) \cdot \log n)$ rounds.
\end{lemma}

We obtain this in three steps. But first, let $\Delta' = \Delta^{2d_i}$ and $d_c= 4 \cdot 3^{4\Delta'}$. Notice that $d_d = 16 \cdot 3^{4\Delta'} = 4d_c$
In the first step, we remove paths and cycles of weak diameter less than $2d_c$ . We can match the nodes in these component such that every node of degree at least $\Delta-i+1$ in $G_i$ is matched expect possibly nodes or neighbor of nodes which were duplicated. This step takes $O(d_c)$ rounds.

\begin{lemma}
\label{lem:weakdiammatching}
    For every path or cycles of weak diameter at most $2d_c$, the algorithms computes a matching such that every unmatched node is of degree strictly less than $\Delta-i+1$ or a duplicated node or a neighbor of a duplicated node. This step is done in $O(d_c)$
\end{lemma}

\begin{proof}
First every node collects its $2d_c$-hop neighborhood in $H$. Then the node of highest ID on paths or cycle of weak-diameter less than $2d_c$ decide arbitrarily on a matching and broadcast it to every nodes in the component. Several cases arise depending on the component.  
For component which are path with even number of nodes or cycles, there is a matching which matches every node of the component. 

Another type of component are paths with an odd number of nodes. We first consider the case where no node was duplicated on this path. In that case, this path $p=(u_1,...,u_k)$ is completely directed, starting from $u_1$ and finishing in $u_k$. Since $u_k$ has no outgoing edge and $p$ contains no duplicated nodes, it can only be a node of degree strictly less than $\Delta-i+1$ in $G_i$.

If the path contains a node which was duplicated, we do the following. Let $p=(u_1,...,u_s,....,u_k)$ the path we consider with $u_s$ a duplicated node. We match the nodes starting from the extremity of the path until we reach the neighbors of $u_s$ as follows. First match $u_1$ to $u_2$,... until we either match $u_{s-3}$ to $u_{s-2}$ or $u_{s-2}$ to $u_{s-1}$. We do the same for the other extremity of the path. Once this matching is done, every node is matched except $u_s$ and possibly $u_{s-1}$ and $u_{s+1}$. If $u_{s-1}$ and $u_{s+1}$ are both matched, then $u_s$ is the only node unmatched. If $u_{s-1}$ or $u_{s+1}$ or both are unmatched, $u_s$ matches with one arbitrarily and at most one of $u_{s-1}$ and $u_{s+1}$ is unmatched. they get matched. 

We also match parts of the components containing nodes of degree $3$. Let $v$ be such a node. If there is a path $p=(v,u_1,...,u_k)$ or a cycle $c=(v,...,u_k,v)$ of weak diameter less than $2d_c$, we match them as follows. Every cycle is matched perfectly if there is one. If the path is oriented from $v$ towards $u_k$, then we match nodes from $u_1$ to wards $u_k$. Since the distance between two duplicated nodes is at least $d_d>8d_c$, $u_k$ is a node of degree strictly less than $\Delta-i+1$ in $G_i$ and so can be left unmatched. If the path is oriented from $u_k$ towards $v$, then we match nodes from $u_k$ to wards $v$ such that every node on the path is matched except possibly $v$. Notice that matching these two paths can be done independently, and might only leave $v$ as unmatched node.

Since every component has weak diameter at most $2d_c$, this steps can be done in $O(d_c)$ rounds.

\end{proof}

We now want to match the nodes left in components of diameter higher than $2d_c$. For this purpose in step $2$, we select nodes in $H$ such that the distance between two selected nodes in $H$ is at least $2d_i$ in $G_i$ and every node in $H$ is at distance at most $d_c$ of a selected node as follows. We first construct graph $H'$ a subgraph of $H$ from which we remove any node at distance less than $2d_i$ of a vertex of degree $\Delta-i+2$ in $G_i$ and from duplicated nodes in $G_i$. We then construct a clustering of the nodes in $H'$, which is described in \BallCreation$(H',d_c)$. We let $H''$ be a copy of $H'$ where we add an edge between two nodes if they are in different clusters and at distance less than $2d_i$ in $G_i$. The maximum degree in the graph $H''$ is at most $\Delta^{2d_i}=\Delta'$. We then apply the algorithm \NodePicking$(H'',\mathcal{C},d_{c}/4,d_{c}/2,\Delta')$ which returns one node for every cluster such that no two selected nodes are neighbors in $H''$.

\begin{lemma}
\label{lem:nodepickmatching}
    In step $2$, the algorithm constructs clusters and then selects nodes in each cluster such that,
    \begin{enumerate}
        \item the distance between selected nodes and nodes of degree $\Delta-i+2$ is at least $2d_i$ in $G_i$,
        \item the distance between two selected nodes is at least $2d_i$ in $G_i$, and
        \item any node is at distance at most $\frac{3}{2}d_c$ of a selected node in $H$.
    \end{enumerate}
    Step $2$ can be performed in $O(f(\Delta) \log n)$ rounds.
\end{lemma}

\begin{proof}
We now are left with nodes on component of weak-diameter at least $2d_c$. We now want to select nodes which we will use in order to obtain an almost perfect matching. Since we don't want these nodes to be close to a node of degree $\Delta-i+2$ in $G_i$, we first compute $H'$, a subgraph of $H$ obtained by removing every vertex at distance at most $2d_i$ from a vertex of degree $\Delta-i+2$ in $G_i$ and from duplicated nodes in $G_i$. 

Now we want to create clusters of sufficiently large size to apply \NodePicking.  To achieve this, we compute an MIS in the power graph $H'^{d_c}$ in $O(d_c+\log^* n) = O(f(\Delta)+\log^* n)$ rounds. We now go back to the graph $H$, and consider the MIS nodes in $H'$ as MIS nodes in $H$.
For each node $v$ in the MIS, we create a cluster $C_v$ containing all of the nodes at distance less then $d_c/4$ of $v$ in $H'$, so every cluster $C_v$ has a diameter at least $d_c/4$ and at most $d_c/2$ since every component is of size at least $2d_c$ . Since the MIS vertices are at at distance greater than $d_c$, these clusters are pairwise disjoint. Let $\mathcal{C}$ be the set containing all the clusters. We let $H''$ be a copy of $H'$ where we add an edge between two nodes if they are in different clusters and at distance less than $2d_i$ in $G_i$. The maximum degree in the graph $H''$ is at most $\Delta^{2d_i}=\Delta'$.
Since $d_c/4 \geq 3^{4\Delta'}$, we can apply \NodePicking$(H'',\mathcal{C},d_c/4,d_c/2,\Deltaball)$ which selects exactly one vertex from each cluster such that the properties in \Cref{lem:node-picking} are respected. This takes at most $O(f(\Delta) \log n)$ rounds. 

First we show property $1$, which states all selected node are at distance greater than $2d_i$ from nodes of degree $\Delta-i+2$ in $G_i$ and nodes of degree $3$ as follows. Every selected node is at distance at most $d_c/4$ of an MIS node. Or MIS nodes are at distance at least $d_c$ of a node of degree $\Delta-i+2$ in $G_i$ and nodes of degree $3$. So selected nodes are at distance at least $3d_c/4 = 8\Delta'^2 \cdot 3^{\Delta'} > \Delta' =\Delta^{d_i} > 2d_i$.

Property $2$ is also satisfied as follows. \NodePicking$(H'',\mathcal{C},d_c/4,d_c/2,\Deltaball)$ makes sure that the distance between selected node is at least $2$ in $H''$, so the distance between selected nodes in $G_i$ is at least $2d_i$.

Property $3$ is also satisfied as every node in $H$ is at most $d_i+d_c$ of an MIS node, and any selected node is at distance at most $d_c/4$ of an MIS node.
\end{proof}

In step $3$, we compute a matching in $H$ using the selected nodes in the previous step. This matching is such that the only unmatched nodes are the selected nodes, the nodes of degree $3$, or one their neighbors. Also, the distance between any two unmatched nodes is at least $2d_i$.

\begin{lemma}
\label{lem:highdiammatching}
    Step $3$ matches every remaining node in $H$ such that the distance between any two unmatched node is at least $d_i$. Step $3$ can be done in $O(d_c)$ rounds.
\end{lemma}

\begin{proof}
We now compute a matching in $H$. 

We first match the neighborhood of nodes of degree~$3$. Since $v$ does not  have any component of diameter less than $2d_c$, there are at least $3$ selected in its $3/2d_c$-hop neighborhood in $H$. We do the following matching. Let $v_k$ a selected node and let $p=(v,v_1,...,v_k)$ the path from $v$ to node $v_k$ such that no $v_j$ is a selected node. We match every node starting from $v_1$ up to $v_k$ or $v_{k-1}$. 

We now match the vertices on the remaining components. In each component, every selected vertex identifies the closest two selected vertex on either side, if they exist. For each pair $s$ and $s'$ of consecutive selected vertices, we computes a matching covering all vertices between them leaving unmatched only the selected vertices and possibly their immediate neighbors as follows. Let $p=(s,s_1,....,s_k,s')$ the path between $s$ and $s'$. Let $s$ be the node of highest ID between $s$ and $s'$. Then $s$ matches every node starting from $s_1$ up to $s_k$ or $s'$. Lastly, we match the nodes which are just on one side of a selected node. Let $p=(s,s_1,....,s_k)$ with $s$ a selected node. We start the matching from $s_k$ up to $s_1$ or $s_2$. Finally, every selected vertex is matched whenever one of its neighbors is still available as follows. Let $s$ a selected node and $s_1$ and $s_2$ its neighbors. If $s_1$ or $s_2$ are left unmatched, $s$ matches with on arbitrarily. Doing all of these matching can be done in $O(d_c)$ rounds.

Consequently, the remaining unmatched node in $H$ are selected node, duplicated nodes, or neighbors of one of these vertices.  In particular, any two unmatched vertices are at distance at least $d_i$ as the distance between any of any two selected node is at least $2d_i$. Any of these step takes $O(d_c)$ rounds.
\end{proof}

We now prove \Cref{lem:spasenodeselec}

\begin{proof}
    The combination of \Cref{lem:weakdiammatching}, \Cref{lem:nodepickmatching} and \Cref{lem:highdiammatching} ensures that every node of degree $\Delta-i+1$ in $G_i$ or more is matched except some nodes which at distance at least $d_i$ between themselves as follows. \Cref{lem:weakdiammatching} ensures that every node of degree $\Delta-i+1$ which is not matched and in a component of size less than $2d_c$ is nodes which was duplicated or is a neighbor of a duplicated node. \Cref{lem:weakdiammatching} ensures that every node of degree $\Delta-i+1$ which is not matched in a component of size greater $2d_c$ is either a node was which duplicated or a selected nodes, or or a neighbor of one of these two types of nodes. The distance between any two duplicated nodes is at least $2d_d > 2d_i$. The distance between any two selected nodes is at least $2d_i$. The distance between a selected node and an unmatched duplicated node or an unmatched neighbor of a duplicated node is at least $2d_i$. And the distance between selected nodes and nodes of degree $\Delta-i+2$ is at least $2d_i$. So the distance between any two unmatched node is at least $d_i$. The overall computation is $O(f(\Delta) \log n)$.
\end{proof}

\subsubsection{Proof of The Main Lemma}\label{sec:mainlemma}

\begin{lemma}
    In phase $i$, the algorithm and colors the edges between matched nodes with color $\Delta-i+2$. The distance between any two nodes of degree $\Delta-i+1$ after this step is $d_i$.
\end{lemma}

\begin{proof}
 Every edge in the matching gets color $\Delta-i+2$. \Cref{lem:spasenodeselec} shows that the distance between two nodes of degree $\Delta-i+1$ which don't get any edge colored is at least $d_i$ and that every node of degree $\Delta-i+2$ gets an edge colored as the only unmatched nodes are at distance at most one of a selected node or a duplicated node, and neither are node of degree $\Delta-i+2$. So, every node in the graph $G_{i+1}$ has a degree of at most $\Delta-i+1$.
\end{proof}

\subsection{Proof of the Main Theorem}
We prove in this section, that \Cref{alg:main-alg} outputs a proper $\Delta+1$ edge coloring. 

\begin{lemma}
    \label{lem:colordegree3graph}
    Let $G$ be a graph of maximum degree $3$ in which the
    nodes of degree $3$ are pairwise at distance at least $5$. Then a proper
    $3$-edge-coloring of $G$ can be computed in $O(\log^* n)$ rounds.
\end{lemma}

\begin{proof}
    Partition the edges into three classes: $E_A$, the edges incident to a node of
    degree $3$; $E_B$, the edges of $E \setminus E_A$ adjacent to no edge of $E_A$;
    and $E_C = E \setminus (E_A \cup E_B)$.

    \emph{Coloring $E_A$.} Every node of degree $3$ properly colors its three
    incident edges with the three distinct colors $1,2,3$. Since degree-$3$ nodes
    are pairwise at distance at least $3$, two such stars are edge-disjoint and
    contain no pair of adjacent edges, so no conflict arises. This takes $O(1)$
    rounds.

    \emph{Coloring $E_B$.} Both endpoints of an edge in $E_B$ have degree at most
    $2$, and by definition no adjacent edge lies in $E_A$. Hence the subgraph
    $(V,E_B)$ is a disjoint union of paths and cycles, and its line graph again
    consists of paths and cycles, which can be properly $3$-colored in
    $O(\log^* n)$ rounds \cite{cole1986deterministic}. Edges of $E_C$ are still
    uncolored at this point, so they impose no constraint.

    \emph{Coloring $E_C$.} An edge of $E_C$ has both endpoints of degree at most
    $2$ and is thus adjacent to at most two edges, so at most two colors are
    forbidden and a color in $\{1,2,3\}$ remains. Moreover, no two edges of $E_C$
    are adjacent: each edge of $E_C$ is adjacent to an edge of $E_A$, so two
    adjacent edges of $E_C$ would place two degree-$3$ nodes at distance at most
    $4$, contradicting the assumption. Hence all edges of $E_C$ can be colored
    greedily and simultaneously in $O(1)$ rounds.
\end{proof}

We now prove the main theorem which we restate for convenience.
\MainEdgeColoringTheorem*

\begin{proof}
    We have $(d_i)$ such that $d_{\Delta-2} = 100$ and for all $\Delta-2\geq i\geq 1$, $d_{i-1} = 16\cdot 11520\cdot 3^{4\Delta ^{d_{i}}}$.
    After $\Delta-2$ phases, all edges expect the edges of $G_{\Delta-1}$ are colored with a color in $(4,\Delta + 1)$. the maximum degree of the nodes in $G_{\Delta-1}$ is at most $3$ and all of the nodes of degree $3$ are at distance $d_{\Delta-2} = 100$ between each other. This graph $G_{\Delta-1}$ can be colored using color in $\{1,2,3\}$ (\Cref{lem:colordegree3graph}).
\end{proof}

We prove additionally the local version of the theorem.

\LocalPaletteTheorem*

\begin{proof}
    The coloring we have obtained is already satisfying these properties. Indeed, in step $i$, every edge colored using the matching has one of its endpoint of degree $\Delta-i+1$. However some edges where colored in small component independently (\Cref{lem:bipedgecolor}). These edges are also colored properly as every bipartite graph can be edge-colored such that every edge $e=(u,v)$ has a color in $\{1,\max(\deg(u),\deg(v))\}$
\end{proof}

\subsection{Node Selection}\label{sec:node-selection}
We prove \Cref{lem:node-picking} in this subsection.

\begin{algorithm}[htb!]
\caption{\NodePicking$(G,\mathcal{C},d_{\max},\Deltaball)$}
\label{alg:node-picking}
\begin{algorithmic}[1]
\State Compute a $2\Deltaball-1$ edge coloring in $G$
\State Let $S=V$ be the set of choosable nodes
\For{$i=1$ to $2\Deltaball-1$}
\State Let $G_i$ the graph induced by the edges of color $i$ in $G$ and the nodes in $S$
\State Group nodes in $G_i$ which are in the same cluster in groups of $3$
\State Let $H_i$ an empty graph and for each group in $G_i$, add a node to $H_i$
\State Add an edge $e'=(u',v')$ in $H_i$ if $u'$ and $v'$ contain two nodes which are adjacent 
\State Apply sinkless orientation on $H_i$
\State Remove from $S$ every node with an incoming edge
\EndFor
\State Output for each cluster a node with only outgoing edges.
\end{algorithmic}
\end{algorithm}

Let $G$ be a bipartite graph of maximum degree $\Deltaball$ and $\mathcal{C}$ a clustering of nodes of the graph with minimum diameter $d_{\min}$ and maximum diameter $d_{\max}$. We want every cluster to select one node such that no neighbor of that node is selected by another cluster. In order to obtain this, at a high level, we  remove nodes from every cluster in a succession of phases. At the end of all phases, every cluster should be left with a small number of nodes such that no neighbors of these nodes are still contained in a cluster.

Initially, the algorithm colors the edges of the graph $G$ with $2\Deltaball-1$ colors which takes $O(\Deltaball+\log^* n)$ rounds. We now do $2\Deltaball-1$ phases from $i=1$ to $2\Deltaball-1$. 

Let $i$ be the current phase and let $G_i$ be the graph induced by the edges of color $i$. From each graph $G_i$, we construct a virtual graph $H_i$. In each graph $G_i$, inside each cluster $C \in \mathcal{C}$, we group nodes in groups of $3$ expect possibly one node of degree less than $3$ and create one node for each of these groups in $H_i$. For every node created in $H_i$, there is an edge between two nodes $v'_i$ and $u'_i$ if one of the nodes in $v'_i$ is neighbor to a node in $u'_i$. Then we compute a sinkless orientation on $H_i$. Every node in $H_i$ has then one outgoing edge if it was a node of degree $3$.
By choosing $d_{\min}$ sufficiently large compared to $\Deltaball$, we can ensure that every cluster has one node with only outgoing edges.

We say a node is \emph{choosable} if it has no incoming edge. We let $S$ be the set of choosable nodes. Initially, all nodes in a cluster are choosable, so $S=V$. 

\begin{observation}
\label{obs:choosablenode}
     Let $s_i$ be the minimum number of choosable nodes in each cluster $C$ at the beginning of phase $i$. Then $s_{i+1} \geq \lfloor s_i/3 \rfloor  - 2$. Each phase $i$ can be done in at most $O(d_{\max}\log n)$.
\end{observation}

\begin{proof}
    Lets assume we are in phase $i$. Let $G_i$ be the  subgraph of $G$ induced by the edges of color $i$ and the nodes in $S$. Let $s_i$ be the minimum number of choosable nodes in each cluster $C$ at the beginning of phase $i$.
    
    We first construct a virtual graph $H_i$. 
    For each cluster $C\in \mathcal{C}$, $C$ groups arbitrarily the choosable nodes inside it, which have an edge with color $i$, in groups of $3$ and creates a node in $H_i$ for every group. Doing this grouping can be done in $2d_{\max}$ rounds, as every cluster can appoint the node of highest ID inside the cluster as a leader and have it decide the grouping arbitrarily. Since the size of $H_i$ is not a multiple of $3$, one group could be of size $2$ or $1$. 
    $H_i$ is the following graph: Every node in $H_i$ has degree at most $3$ and represent a group of nodes, all inside the same cluster. Since we want to simulate the sinkless orientation algorithm, we remove from $H_i$ the nodes of degree strictly less than $3$ and create a half edge for each edge removed which had an endpoint of degree $3$. Each node in $H_i$, gets for ID, the ID of the node of highest ID in its group.
    Now the algorithm simulates computation of sinkless orientation in the graph $H_i$ (\Cref{lem:sinklesso}), which is possible as every node in $H_i$ now has a degree of exactly $3$. The simulation of sinkless orientation in $H_i$ can take up to $O(d_{\max}\log n)$ rounds as the distance between any two neighbors in $H_i$ is at most $d_{\max}$ in $G$.
    Every node in $H_i$ with degree $3$ has at least one outgoing edge. So, every node in $H_i$ with degree $3$ contains one node which remains choosable.
    
    Lets at a cluster $C\in \mathcal{C}$. In the beginning of phase $i$, the number of choosable nodes was $s_i$. A node went from being choosable to not-choosable if it had an edge oriented as incoming during the sinkless orientation computation. The number of nodes in $H_i$ which were in $C$ was at most $\lfloor s_i/3 \rfloor$ + 2. And every one of these nodes, if they had degree $3$ had an outgoing edge. So at most $2\lfloor s_i/3 \rfloor+2$ of the nodes in $C$ go from choosable to non choosable. Which means that number of choosable nodes in $B$ has decrease to at least $s_i - (2\lfloor s_i/3 \rfloor+2) = \lfloor s_i/3 \rfloor  - 2$.
\end{proof}

We now prove the main lemma. 

\NodeSelection*

\begin{proof}
    We have $d_{\min} \geq 3^{4\Deltaball}$. We obtain the result by the following induction, that $s_i \geq 3^{4\Deltaball + 2 - 2i}$. This is true for step $1$ as $s_1\geq d_{\min}$. Let assume this is correct until step $i$. 
    Then, by \Cref{obs:choosablenode}, $s_{i+1} \geq \lfloor s_i/3 \rfloor  - 2 \geq \lfloor3^{4\Deltaball - 2i}/3\rfloor - 2 = 3^{4\Deltaball - 2i -1} -2 \geq 3^{4\Deltaball - 2i - 2}$.
     After the $2\Deltaball -1$ phases, the number of choosable nodes is $d_{2\Delta} = 3^{2} = 9$.
     The overall algorithm has a round complexity of $O(\Delta d_{\max} \log n)$.
\end{proof}

\section*{Acknowledgements}
All results were obtained by the authors without the use of AI. For the write-up, GPT-5.5 and Claude Opus 4.8 have partially been used for literature review and language editing and have assisted in drafting preliminary versions of technical and related-work sections. All generated content has been carefully checked and modified by the authors. The authors assume responsibility for all content.

Funded by the European Union. Views and opinions expressed are however those of the author(s) only and do not necessarily reflect those of the European Union or the European Research Council. Neither the European Union nor the granting authority can be held responsible for them. This work is supported by ERC grant \href{https://doi.org/10.3030/101162747}{OLA-TOPSENS} (grant agreement number 101162747) under the Horizon Europe funding programme.

\printbibliography
\end{document}